%% file: main.tex
\documentclass{article}
\input{header.tex}

\usepackage[minbibnames=2]{biblatex} 
\usepackage[intoc]{nomencl}

\makenomenclature
    \nomenclature{$A^\adj$}{Pseudo-adjoint}
    \nomenclature{$A^\t$}{Transpose}
    \nomenclature{$\R^{n\times n}_{UT}$}{Strictly upper-triangular real $n\times n$ matrices}
    \nomenclature{$A^{(UT)}$}{Strictly upper-triangular part}
    \nomenclature{$x\cdot y, \<x,y\>$}{Euclidean inner product}
    \nomenclature{$\<f,g\>_{\mathcal B}$}{Pseudo-Boolean inner product}
    \nomenclature{$\|x\|$}{Euclidean norm}
    \nomenclature{$[T,b]$}{Affine map $x \mapsto Tx + b$}
    \nomenclature{$\Aff(V,W)$}{Set of affine functions $V\to W$}
    \nomenclature{$\Sigma^n$}{Hypercube $\{-1,1\}^n$}
    \nomenclature{$(n,m,f)$}{Circuit representing the Boolean function $f:\Sigma^n \to \Sigma^m$}
    \nomenclature{$x^{\otimes2}$}{Outer square}
    \nomenclature{$B_r(x)$}{Closed Hamming ball of radius $r$ with center $x$}
    \nomenclature{$[n]$}{The set $\{1, 2, \dots, n\}$}
    \nomenclature{$\#S$}{Cardinality of $S$}
    \nomenclature{$B^n$}{Boolean cube $\{0,1\}^n$}
    \nomenclature{$\mathbb E[f]$}{Expected value}
    \nomenclature{$\mathcal{B}^n$}{Space of pseudo-Boolean functions $\Sigma^n \to \R$}
    \nomenclature{$f\times g$}{Product function $x \mapsto (f(x),g(x))$}

\title{Geometric Theory of Ising Machines}
\author{Andrew Moore, Zachary Richey, Isaac Martin}

\begin{document}
\maketitle

\begin{abstract}
    We contribute to the mathematical theory of the design of low temperature Ising machines, a type of experimental probabilistic computing device implementing the Ising model. Encoding the output of a function in the ground state of a physical system allows efficient and distributed computation, but the design of the energy function is a difficult puzzle. We introduce a diagrammatic device that allows us to visualize the decision boundaries for Ising circuits. It is then used to prove two results: (1) Ising circuits are a generalization of 1-NN classifiers with a certain special structure, and (2) Elimination of local minima in the energy landscape can be formulated as a linear programming problem. 
\end{abstract}

\section{Introduction}

The `Ising computer' lies at the confluence of physics, artificial intelligence, and quantum-inspired classical computing. The dream is simple: if we could create a physical system whose ground states were the answers to a computational problem, then the tendency of the system to minimize its energy would solve the problem autonomously. This dangles the tantalizing promise of cutting computational power consumption by several orders of magnitude and circumventing the obstacles posed by the traditional Von Neumann computing architecture, an increasingly relevant problem in the AI age \cite{Melanson2025}. The Ising model, a favorite of statistical physics, offers a natural starting place: indeed there is a thriving field of hardware research implementing programmable Ising-type systems using everything from coupled analog oscillators \cite{roy2025experimentsoscillatorbasedising, Chou_2019} to magnetic tunnel junctions \cite{yang2025250magnetictunneljunctionsbased} to entangled photons in a loop of fiber optic cable \cite{Yamamoto2017, McMahon_2016}. A good review of the current implementation technologies and the desirability of Ising computing from a combinatorial optimization perspective can be found in \cite{mohseni2022isingmachineshardwaresolvers}. 

The problem of designing the energy functions for such circuitry, however, has received far less attention. Most research either simply re-implements traditional digital logic using Ising spins \cite{tsukiyama2024designingunitisingmodels}, or restricts the focus to combinatorial optimization problems which naturally lend themselves to formulation as Ising Hamiltonians \cite{Lucas_2014}. Very little work has been done to characterize how more general circuitry might be modeled, and what the expressive limits of the system might be. The problem of optimizing positive temperature Ising circuitry has been investigated with AI methods by our collaborators in \cite{knoll2024solvingboltzmannoptimizationproblems}. Significant progress has been made on the `inverse Ising problem' for modeling the probability distributions of interacting systems \cite{Nguyen_2017}, but this approach does not consider the problem of prescribing qualitative limiting behavior, i.e. the desired ground states. 

This paper analyzes the geometry and expressive power of Ising circuits in a non-statistical setting: rather than maximizing probabilities at finite temperature, we only insist that our circuits produce correct answers in the zero temperature limit. This weaker condition is the absolute minimum requirement for the circuit to `compute' the desired function, under the assumption that the output is observed by observing the state of the system at a given point in time. Furthermore, it reduces the problem to the design of a feasible system of linear inequalities. We will present a useful diagrammatic tool for representing the decision surfaces of such low-temperature Ising circuits. Using this visual tool, we will prove two main results:
\begin{enumerate}
    \item Ising circuits are a mild generalization of 1-NN classifiers whose centroids are the vertices of a parallelepiped. The output-output interactions, eliminated in the more popular CRBM (Conditional Restricted Boltzmann Machine), allow for nonlinear classification, though only slightly. Our diagrammatic approach allows us to visualize the nonlinear decision surfaces as an affine transformation of a Voronoi diagram of a parallelepiped, making the connection to 1-NN classifiers. The correspondence is then formally proven. 
    \item A slight modification to the linear system completely eliminates all parasitic local minima of the energy function. Local minima are a plague on all energy-based computation systems. Typical solutions include annealing the system, probabilistically hoping to avoid becoming trapped in a false minimum. We show that the requirement that no incorrect states be a local minimum can be expressed, with minimal extra cost, within the linear programming framework, using a derivation from our diagrammatic device. Solutions to this modified problem give circuits which are guaranteed to converge to the correct solution very quickly under Glauber dynamics at low temperature; they therefore have all the advantages of a convex energy landscape while being less restrictive. We supply a mathematical proof that our modified linear program indeed yields no-local-minima energy landscapes, and suggest and algorithm for its implementation. 
\end{enumerate}

The Ising circuit that we are discussing is very similar to the `conditional Boltzmann machine', though considered in the $\beta \rightarrow \infty$ limit. That model is well studied in the probabilistic setting of modeling probability distributions. We will stick with the language of `Ising circuit', as the design of small systems is of greater interest to probabilistic computing hardware than to artificial intelligence. However, the reader may read ``zero temperature conditional Boltzmann machine" for ``Ising circuit'' throughout. Nonetheless, this work has a very different flavor from the literature on CBMs. In fact, this work is almost entirely self-contained. For context, we have cited several results from our previous algorithmically focused paper on the same subject \cite{andrewisaac}. 

\subsection{Outline}

We will begin by introducing the formal definitions of Ising circuits in section \ref{sec:isingcircuits}, along with some basic results. Next, in section \ref{sec:res_ham}, we introduce a relaxed formulation and a diagrammatic device for analyzing Ising circuits. Then, we proceed to the two main results in sections \ref{sec:superset} and \ref{sec:localminima}. We have pushed a number of elementary facts, definitions, and notation regarding linear algebra and boolean functions to appendices \ref{sec:linearalgebra} and \ref{sec:boolean}. Finally, we include some interesting remarks on our empirical review of small Ising circuits in appendix \ref{sec:empirical} and a connection to tropical geometry in appendix \ref{sec:tropical}. 

\tableofcontents

\section{Ising Circuits}
\label{sec:isingcircuits}

Ising circuits are physical Ising spin systems which are thought of as computing devices: we fix the states of certain spins, called the inputs, and observe the state of other spins, called the outputs. If the temperature is low, the system will be in its ground state (conditional on the fixed state of the input spins) with high probability; i.e. it will have probabilistically computed a function of the input spins. It is necessary to bring in some standard definitions from statistical physics:

\begin{definition}
    Define $\Sigma = \{-1,1\} \subset \R$ and $B = \{0,1\} \subset \R$.
    The $n$-dimensional \textit{hypercube} is the set $\Sigma^n$. We may also refer to the convex hull of $\Sigma^n$ or the hypercube graph as the hypercube, depending on context. An element $s \in \Sigma^n$ is called a \textit{spin state} or \textit{spin configuration}. It represents a list of spin values of $n$ particles in an Ising system.
\end{definition}

\begin{definition}[Ising Hamiltonian]
    An \textit{Ising Hamiltonian} is a quadratic pseudo-Boolean polynomial $H: \Sigma^d \to \R$ with no constant terms.
    Any Ising Hamiltonian can be written as
    \begin{align}
        H(s) = h\cdot s + s^\t Js = \sum_{i=1}^d h_i s_i + \sum_{1 \le i < j \le d} J_{ij} s_is_j,
    \end{align}
    for some vector $h \in \R^d$, called the \textit{local biases}, and some strictly upper-triangular matrix $J \in \R^{d\times d}_{UT}$, called the \textit{coupling coefficients} or \textit{interaction strengths}.\footnote{Alternatively, we could replace $J$ with its symmetric part:
    \begin{align}
        \text{Sym}(J) = \frac{1}{2}(J + J^\t)
        = \frac{1}{2}\begin{pmatrix}
            0 & J_{12} & J_{13} & \dots \\
            J_{12} & 0 & J_{23} & \dots \\
            J_{13} & J_{23} & 0 & \dots \\
            \vdots & \vdots & \vdots & \ddots
        \end{pmatrix},
    \end{align}
    since $s^\t Js = s^\t \sym(J)s$.
    We note $\sym(J)$ is a symmetric hollow matrix, and also the Hessian of $H(s)$.}
    $H(s)$ represents the energy of an Ising system in the spin state $s$ with parameters specified by $h$ and $J$. We say that an Ising Hamiltonian $H$ is \textit{generic} when all the energy levels are distinct, i.e. $H(x) = H(y)$ iff $x = y$. 
    
    
\end{definition}

\subsection{Realizing Computations with Ising Systems}

It will be very useful for our analysis to separate the notion of what a circuit should do from its actual physical manifestation. An Ising system (set of spins and a Hamiltonian) does not by itself have the interpretation of a circuit. We formalize this by calling a `circuit' the abstract prescription of a function which maps inputs to outputs, and interpreting an Ising system based on the circuit that it is meant to implement or realize:

\begin{definition}[Circuit]
    A \textit{circuit of shape $(n,m)$} is a triple $(n,m,f)$, where $n$ and $m$ are natural numbers and $f: \Sigma^n \to \Sigma^m$ is a Boolean function.     If an Ising Hamiltonian $H: \Sigma^{n+m} \to \R$ is associated to a circuit $(n, m, f)$, we decompose its input as $s = (x, y) \in \Sigma^n \times \Sigma^m$.  Here, $x \in \Sigma^n$ represents a particular input spin configuration, or \textit{input state}, and $y$ represents an \textit{output state}.
\end{definition}

We will assume that we have total control over fixing the input state, i.e. that they can be pinned to any binary pattern. Now, if we consider $n+m$ spins, with the first $n$ spins pinned at some choice of states, we can see that any generic\footnote{See \Cref{prop:generic} for technical details on uniqueness of the argmin.} Ising Hamiltonian $H$ executes the computation of the following function on $x \in \Sigma^n$:
\begin{align}
    f_H(x) := \argmin_{y\in\Sigma^m} H(x,y)
\end{align}
Where by `executes the computation of $f_H$' we mean that with respect to the Boltzmann-Gibbs distribution on $\Sigma^{n+m}$ given by $H$,
\begin{align}
    \lim_{\beta \rightarrow \infty} \mathbb{P}[(\sigma_{n+1}, \dots, \sigma_{n+m}) = f_H(x) | (\sigma_1, \dots, \sigma_n) = x] = 1
\end{align}
We say that $H$ \textit{encodes} the circuit $(n, m, f_H)$. We can now state precisely what it means for a generic Ising Hamiltonian to embody a particular abstract circuit: $H$ encodes $(n,m,f)$ means $f  = f_H$. This marks the last time that we will mention temperature: at positive temperature, Ising systems compute $f_H$ probabilistically; however, at zero temperature they are deterministic. $f_H$ therefore represents the ideal behavior which is approximated by a positive-temperature circuit, and we will focus on this ideal behavior.

The basic problem of Ising circuit design is this: \textit{Given a circuit, how do we find a Hamiltonian that encodes it?} In practice, furthermore, this usually is not possible. Therefore, we must often ask a more basic question, to wit: \textit{Given a circuit, does there exist a Hamiltonian that encodes it?}. If there does exist such a Hamiltonian, we call the circuit \textit{feasible}, otherwise, it is \textit{infeasible}. Both questions can be answered at once: they are reducible to a linear programming problem.

\begin{lemma}[Linearity of the Inverse Problem]
Finding a Hamiltonian $H$ which encodes $(n,m,f)$ is equivalent to solving the linear programming problem
\begin{align}
\label{eq:linprog}
    \min_u \|u\|_1 \text{ s.t. } \langle u, v(x,y) -v(x,f(x))\rangle \geq 1 \forall y \in \Sigma^m\setminus\{f(x)\} \forall x \in \Sigma^n
\end{align}
Where $v(s) = \text{Vec}[s,(s^{\otimes 2})_{UT}] \in \Sigma^{n+m+{n + m \choose 2}}$
\begin{proof}
    Observe that when $u$ is the vector of coefficients of a Hamiltonian $H$, $H(s) = \langle u, v(s)\rangle$. Furthermore, by definition
    \begin{align}
        f(x) = \argmin_{y\in\Sigma^m} H(x,y) \iff H(x,y) > H(x, f(x)) \forall y \in \Sigma^m\setminus\{f(x)\} \forall x \in \Sigma^n
    \end{align}
    Since there are a finite number of constraints, we may let
    \begin{align}
        \epsilon(u) := \min\{ \langle u, v(x,y) -v(x,f(x)) \rangle | y \in \Sigma^m\setminus\{f(x)\}, x \in \Sigma^n\}
    \end{align}
    Observe that $(n,m,f)$ encoded by a Hamiltonian $H$ with coefficient vector $u$ iff $\epsilon(u) > 0$. Note that $u$ solves \Cref{eq:linprog} when $\epsilon(u) \geq 1$. If $\epsilon(u) \geq 1$, clearly $\epsilon(u) > 0$. On the other hand, if $\epsilon(u) > 0$, then $\epsilon(u/\epsilon(u)) = 1$, so $u/\epsilon(u)$ solves \Cref{eq:linprog}. Therefore $(n,m,f)$ is feasible iff \Cref{eq:linprog} is feasible, and up to linear rescaling, the set of $H$ which encodes $(n,m,f)$ is precisely the feasible region of \Cref{eq:linprog}. 
\end{proof}
\end{lemma}

This lemma is fundamental for the rest of our analysis, much of which is based on the manipulation of the linear constraint set. While the linear program has exponentially many constraints, usually making it impractical for concretely designing circuits with more than 20 output spins even with modern computers, it is very useful as an analytic tool.

The objective of the linear program, which we have set as $L^1$ norm minimization of the coefficient vector, is only partially arbitrary. Of course, when the problem is feasible, there are many solutions. However, in many hardware applications, the $L^\infty$ norm is limited, the $L^1$ norm measures total energy consumption, or there is some other such restriction on the magnitude of the coefficients. We have chosen $L^1$ minimization on the general principle that it balances concerns about minimizing the coupling sizes and number of nonzero couplings, but in truth it may be substituted with any linear objective without changing the feasibility results that we are concerned with.

The following easy but important fact states that feasible one-output Ising circuits are precisely \textit{threshold functions}, in which case the linear program \Cref{eq:linprog} reduces to the $L^1$-regularized SVM. 
\begin{theorem}\label{cor:threshold}
    A circuit $(n,1,f)$ is feasible if and only if $f$ is a threshold function.
    \begin{proof}
        Let $H:\Sigma^n\times\Sigma^m\to\R$ be a generic Ising Hamiltonian, which we may assume is reduced by the procedure in \Cref{subsubsec:reduction}.
        Since there is only one output, there are no coupling coefficients, and hence $J = 0$. Thus, the $H$ is simply $H(x,y) = A(x)y$
        for some $A \in \Aff(\R^n, \R)$. We see that $y = -1$ is the ground state when $A(x) > 0$, and $y = 1$ is the ground state when $A(x) < 0$. Therefore, the above Hamiltonian encodes $(n,1,f)$ if and only if $f(x) = -\sgn(A(x))$. 
        Since $A$ is affine, $A(x) = w_0 + w\cdot x$ for some $w_0 \in \R$ and $w \in \R^n$.  Thus, $f(x) = \sgn(-w_0 - w\cdot x)$.
    \end{proof}
\end{theorem}

This result shows that Ising circuits in general can be thought of as a certain multiple-output generalization of a support vector machine, where the nonlinearity is not introduced by a kernel but by the quadratic coupling between the outputs. This means that such devices are more expressive than linear models (i.e. traditional support vector machines) but not as much as a fully nonlinear classifier. The actual geometry of the nonlinear decision surfaces will be visualized in section \ref{sec:res_ham} and analyzed in section \ref{sec:superset}.

\subsubsection{Simplifying the Hamiltonian}
\label{subsubsec:reduction}

An astute reader may have noticed that the constraints of the linear programming problem do not depend on the variables which correspond to terms in the Hamiltonian only dependent on input spins, since $v(x,y) - v(x,f(x))$ has many zero entries. This redundancy can be removed: For any computation with an Ising system, the input spins will be ``pinned," so that their states are not affected by the other spins in the system. Thus, we can freely assume that the local biases of the input spins and the coupling coefficients between them are all zero.\footnote{This choice is not always desirable when chaining multiple circuits together, but it is perfectly fine for individual circuits.}  Thus, the Ising Hamiltonian reduces to

    \begin{align}
       H(x,y) &= \sum_{i=1}^m h_i y_i + \sum_{k=1}^n \sum_{i=1}^m J_{x_ky_i}x_ky_i + \sum_{1\le i < j \le m} J_{y_iy_j}y_iy_j \\
       &= \sum_i \left( h_i + \sum_k J_{x_ky_i}x_k\right)y_i + \sum_{i<j} J_{y_iy_j} y_iy_j.
    \end{align}
    We define
    \begin{align}
        A_i(x) := h_i + \sum_k J_{x_ky_k}x_k \in \Aff(\R^n,\R), \\
        A := A_1\times \dots \times A_m \in \Aff(\R^n, \R^m).
    \end{align}
    Thus, an Ising Hamiltonian is given by
    \begin{align}
        \label{hamiltonian}
        H(x, y) = A(x)\cdot y + y^\t Jy = \sum_i A_i(x)y_i + \sum_{i<j} J_{ij} y_iy_j,
    \end{align}
    for some $A \in \Aff(\R^n,\R^m)$ and $J \in \R^{m\times m}_{UT}$.

    It should be noted that the ability to pin the input spins is functionally equivalent to the ability to set their local bias terms. If their local bias terms are set strongly enough relative to the interaction strengths of the system, the qualitative structure of the energy levels will be equivalent. We choose to regard them as pinned because it is simpler, but this modeling choice in practice is hardware-dependent. 

\subsubsection{Genericity}

There is another loose end that we need to tie up. Earlier, we ignored all degeneracies by assuming that we were working with generic Ising Hamiltonians. This was needed to make the argmins unique. It is clear, of course, that a Hamiltonian which encodes a function $f$ cannot have any degeneracy in its ground state, conditional on a fixed input $x$. However, the Hamiltonian given by the solution to the linear programming problem \Cref{eq:linprog} may not be generic. To fix this, we will show that we can always perturb a non-generic Hamiltonian into a generic one which encodes the same circuit. 

\begin{lemma}\label{lem:generic}
\label{prop:generic}
There exists $H$ that encodes $(n,m,f)$ iff there exists a generic $H$ that encodes $(n,m,f)$. 
\begin{proof}
    Assume $H$ encodes $(n,m,f)$. We will show that there exists a generic Hamiltonian encoding the same circuit. The method is as follows: given $H$ with degeneracies, we will perturb it to remove one degeneracy at a time. Inductively, this shows that there must be a generic solution to $f$, since the total number of degeneracies is finite. We define the degeneracy number $D(H) = |\{(\alpha, \beta) \in \Sigma^{n+m} \times \Sigma^{n+m} : \alpha \neq \beta, H(\alpha) = H(\beta)\}|$. 

    \noindent
    Assume $D(H) > 0$. Then there exists $a, b \in \Sigma^{n+m}$ such that $H(a) = H(b)$ and $a \neq b$. Define the solution energy gap
    \begin{align}
        \delta := \min_{x \in \Sigma^n} \min_{y \in\Sigma^m \setminus \{f(x)\}} H(x,y) - H(x, f(x))
    \end{align}
    It follows from the definition of encoding that $\delta > 0$. Now, define the minimum energy gap 
    \begin{align}
        \epsilon := \min_{\alpha, \beta \in \Sigma^{n+m}, H(\alpha) \neq H(\beta)} |H(\alpha) - H(\beta)|
    \end{align}
    Again by definition, $\epsilon > 0$, and $\epsilon < \delta$. Now, pick any Hamiltonian $R$ such that $R(a) \neq R(b)$. Let $|R|$ denote $\max_x |R(x)|$. Consider 
    \begin{align}
        S := \frac{\epsilon}{3|R|}R
    \end{align}
    And let $\hat{H} := S + H$. Evidently, $\hat{H}(a) \neq \hat{H}(b)$, so we have removed a degeneracy. Now, consider any $w, z \in \Sigma^{n+m}$ such that $H(w) \neq H(z)$. Assume without loss of generality that $H(w) < H(z)$. Then $\hat{H}(w) \leq H(w) + \epsilon/3 < H(z) - \epsilon/3 \leq \hat{H}(z)$, i.e. no new degeneracy has been created. It follows that $D(\hat{H}) < D(H)$. Since this argument also applies to $(x,y)$ and $(x, f(x))$, and $\epsilon < \delta$, $\hat{H}$ still encodes $(n,m,f)$. Finally, $D(H)$ is finite, since it's the cardinality of a finite set. Therefore by iterating this process, we must eventually create a generic Hamiltonian encoding $f$. 
\end{proof}
\end{lemma}

 It follows that we may always freely assume that a Hamiltonian encoding a circuit is generic, making all minimizers over subsets of energy levels uniquely defined, which will be frequently useful. This should put to rest any fears the reader may still be retaining regarding the uniqueness of the arg-minimizers. 

\subsection{Auxiliary Spins}

We have seen that Ising circuits are fairly restrictive in terms of the functions that they can realize: while better than linear classifiers, they only have access to the nonlinearity provided by quadratic output-output couplings. It is clear, therefore, that most interesting functions will not be feasible (for example, see \Cref{fig:xor_x1}). However, this issue can be alleviated by adding new `hidden' or \textit{auxiliary} spins, introducing new degrees of freedom into the model. Since we are concerned with modeling a function $f$ and are not worried about the actual value of these extra spins, it makes sense to define a relaxed notion of feasibility:

\begin{definition}[Feasible with auxiliaries]
    A circuit $(n,m,f)$ is \textit{feasible with $k$ auxiliaries} if there exists an Ising Hamiltonian $H: \Sigma^n \times \Sigma^m \times \Sigma^k \to \R$ such that, for each $x \in \Sigma^n$
    \begin{align}
        (f(x), \tilde z) = \argmin_{(y,z) \in \Sigma^{m+k}} H(x,y,z),
    \end{align}
    For some $\tilde{z}$. 
    In other words, if $(x,y,z)$ is the lowest-energy state among all states with the same input, then $y$ must be the correct output.
    We say that such an $H$ \textit{encodes} the circuit $(n,m,f)$. As before, we call $x \in \Sigma^n$ the \textit{input state} and $y \in \Sigma^m$ the \textit{output state}.  We call $z \in \Sigma^k$ the \textit{auxiliary state}.
\end{definition}

Recall that our second fundamental question was ``Given a circuit, does there exist a Hamiltonian that encodes it?'' Without auxiliaries, the answer was ``only if \Cref{eq:linprog} is feasible.'' With auxiliaries, the answer is always yes. This result is constructive and follows immediately from the theory of polynomial quadratization; see \cite[Proposition 2.1]{andrewisaac}. In this paper, we are mostly interested in what functions are computable with a fixed number of spins. Unfortunately, the indeterminacy of $\tilde z$ means that the problem of finding $H$ is now no longer a linear programming problem but a mixed integer-linear programming problem, since we need to specify $\tilde z$ so make the problem of finding $H$ linear. We can make this precise as follows:

\begin{proposition}\label{prop:aux_are_outs}
    A circuit $(n,m,f)$ is feasible with $k$ auxiliaries if and only if there exists a function $g:\Sigma^n \to \Sigma^k$ such that $(n,m+k,f\times g)$ is feasible without auxiliaries.  We call such a $g$ an \textit{auxiliary map}.
    \begin{proof}
        Let $H(x,y,z)$ be an Ising Hamiltonian that encodes $f$. We may assume that $H$ is generic by \Cref{prop:generic}. Select
        \begin{equation}
            g(x) = \argmin_{z\in\Sigma^k} H(x,f(x),z).
        \end{equation}
        Fix $x \in \Sigma^n$.  Suppose
        \begin{equation}
            (\widetilde y, \widetilde z) \in \argmin_{(y,z) \in \Sigma^{m+k}} H(x,y,z).
        \end{equation}
        Since $H$ encodes $f$, $\widetilde y = f(x)$.  Thus,
        \begin{align}
            \widetilde z \in \argmin_{z \in \Sigma^k} H(x,f(x),z)
            = g(x).
        \end{align}
        Therefore, $(\widetilde y, \widetilde z) = (f\times g)(x)$, which implies $H$ encodes $f\times g$. Conversely, suppose $(n,m+k,f\times g)$ is feasible.  Then it is encoded by an Ising Hamiltonian $H(x,y,z)$.  Thus, for each $x\in\Sigma^n$, if
        \begin{equation}
            (\widetilde y, \widetilde z) \in \argmin_{(y,z) \in \Sigma^{m+k}} H(x,y,z),
        \end{equation}
        then $(\widetilde y, \widetilde z) = (f\times g)(x)$.  In particular, $\widetilde y = f(x)$, which implies $H$ encodes $f$.
    \end{proof}
\end{proposition}

It follows that we can treat auxiliary spins as extra outputs without loss of generality. For this reason, most of this document considers circuits with no auxiliary spins; if we choose to ignore some of the outputs, that is a special case. The problem of finding an auxiliary map with $k$ or fewer spins which renders some function $f$ feasible is in general very difficult, and out of the scope of the current paper, though we have previously published algorithms and empirical results regarding this problem in \cite{andrewisaac}.

\section{Residual Hamiltonian Diagrams}\label{sec:res_ham}

The main approach of this paper is based on a decoupling of the Hamiltonian into the input-output interactions and the output-output interactions. Since inputs are pinned, the input-output interactions are effectively affine, while the true nonlinearity arises from the output-output interactions. It is easy to see, for example, that setting all output-output coupling strengths to zero results in an ensemble of linear support vector machines. We would like to visualize and isolate the nonlinear separating behavior from the affine influence of the pinned inputs. To do this, we introduce a new energy function in which the affine influence is freely chosen:

\begin{definition}[Residual Ising Hamiltonian]    
    \label{def:res_ham}
    Given $J \in \R^{m \times m}_{UT}$, its \textit{residual Ising Hamiltonian} is given by
    \begin{align}
        \label{resid_ham}
        E_J(a, y) = a\cdot y + y^\t Jy = \sum_i a_i y_i + \sum_{i < j} J_{ij} y_i y_j.
    \end{align}
    \end{definition}

    \noindent 
    $E_J$ is the same as the original Hamiltonian given in \eqref{hamiltonian}, except we treat the $A_i$ as inputs \textit{per se}, rather than affine functions of the input spins. $E_J$ is associated to a \textit{residual Ising solution map} or \textit{ground state map}
    \begin{align}\label{eqn:ground-state-map}
        M_J(a) := \argmin_{y \in \Sigma^m} E_J(a,y).
    \end{align}
    In other words, given $J$ and $a$, $M_J(a)$ is the output state (or set of output states) that minimizes the residual Hamiltonian. For an Ising system whose Hamiltonian is given by \eqref{hamiltonian}, $M_J(A(x))$ the output state in which we expect to observe the system when the input is fixed at $x$ (at least in the low temperature limit). 

    \begin{figure}[h]
    \centering
    \includegraphics[width = 0.25\textwidth]{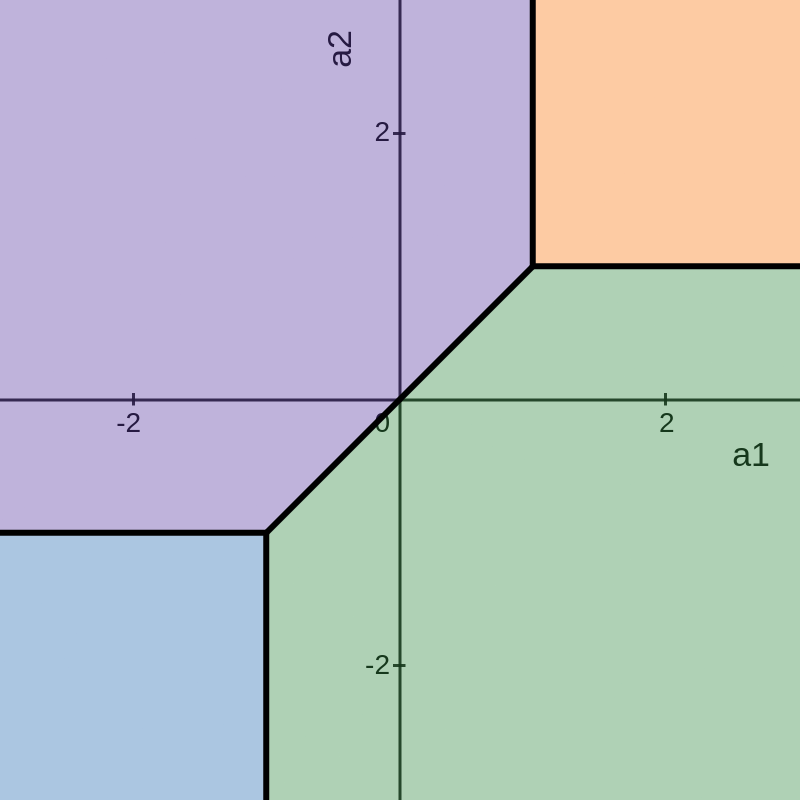}
    \quad
    \includegraphics[width = 0.25\textwidth]{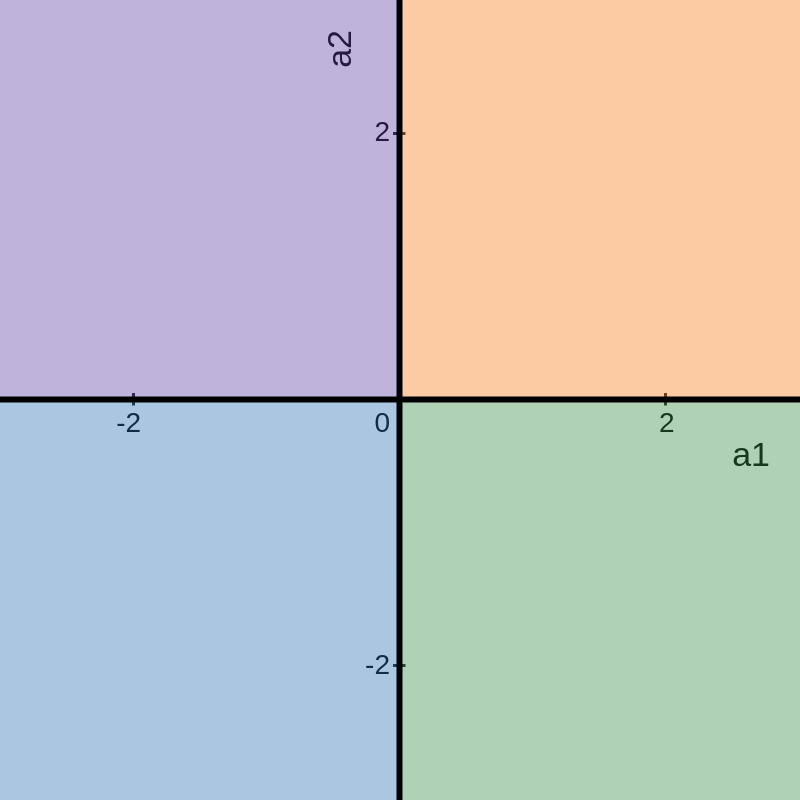}
    \quad
    \includegraphics[width = 0.25\textwidth]{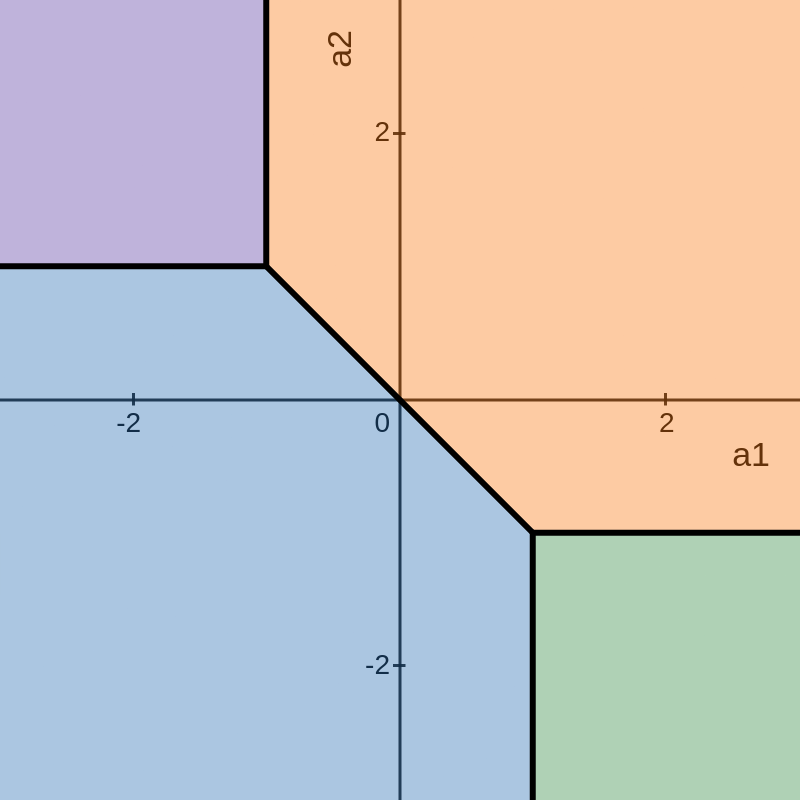}

    \caption{The minimizing partition $\mathcal{C}_J$ for a 2-bit-output circuit, plotted in $a$-space.
    The graphs are for $J_{12} = 1, 0,$ and $-1$, respectively. The colors are as follows: Orange: $\mathcal{C}_J((-1,-1))$ \quad Purple: $\mathcal{C}_J((1,-1))$ \quad Blue: $\mathcal{C}_J((1,1))$ \quad Green: $\mathcal{C}_J((-1,1))$ \quad Black: $B_J$}
    \label{fig:a_space}
\end{figure}
    
    Additionally, for $y \in \Sigma^m$ we define its \textit{minimizing cell} as
    \begin{align}
        \mathcal{C}_J(y) &:= \{a \in \R^m : M_J(a) = y\}\\
        &= \{a \in \R^m : E_J(a,z) > E_J(a,y),\ \forall z \in \Sigma^m\setminus\{y\}\}.
    \end{align} $\mathcal{C}_J(y)$ represents the region in $a$-space for which the output $y$ is the ground state.
    We define the \textit{boundary set} to be
    \begin{align}
        B_J := \{a \in \R^m : \#M_J(a) > 1\}.
    \end{align}
    That is, $B_J$ is the region where there are two or more ground states, and it is also the boundary between the minimizing cells.
    It follows that
    \begin{align}
        \mathcal{C}_J := \{\mathcal{C}_J(y) : y \in \Sigma^m\} \cup \{B_J\}
    \end{align}
    is a partition of $\R^m$ called the \textit{residual Ising minimizing partition}. This partition gives us a visualization of the nonlinear decision boundaries of an Ising circuit; for example, the case of two output spins is shown in \Cref{fig:a_space} and three output spins in \Cref{fig:3_output}.

Now, we can reformulate our discussion of feasibility and the influence of auxiliary spins in terms of these diagrams, where the affine affine part of the original Hamiltonian is conceptualized as an affine map of the input hypercube into the residual Ising minimizing partition:

\begin{SCfigure}[100]
    \includegraphics[width=0.25\textwidth]{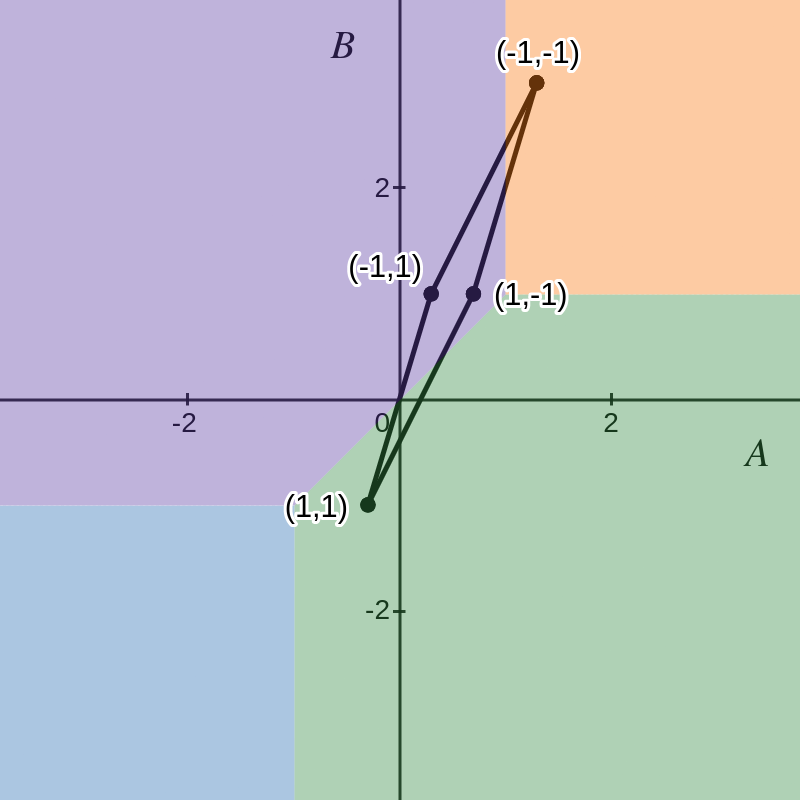}
    \caption{A visual proof of the feasibility of the circuit
    \begin{align}
        f(x_1, x_2) = (\textsf{XOR} \times \textsf{AND})(x_1,x_2).
    \end{align}
    The figure depicts the image of the input square $\Sigma^n$ under the affine map
    \begin{align}
        A(x) = \begin{pmatrix} -0.3 & -0.5 \\ -1 & -1 \end{pmatrix}
        \begin{pmatrix} x_1 \\ x_2 \end{pmatrix}
        + \begin{pmatrix} 0.5 \\ 1 \end{pmatrix}.
    \end{align}
    We have colored $a$-space according to the minimizing partition $\mathcal{C}_J$ with $J_{12} = 1$. Observe that each input $x$ is mapped to $\mathcal{C}_J(f(x))$.}
    \label{fig:xor_soln}
\end{SCfigure}

\begin{SCfigure}[100]
    \includegraphics[width=0.25\textwidth]{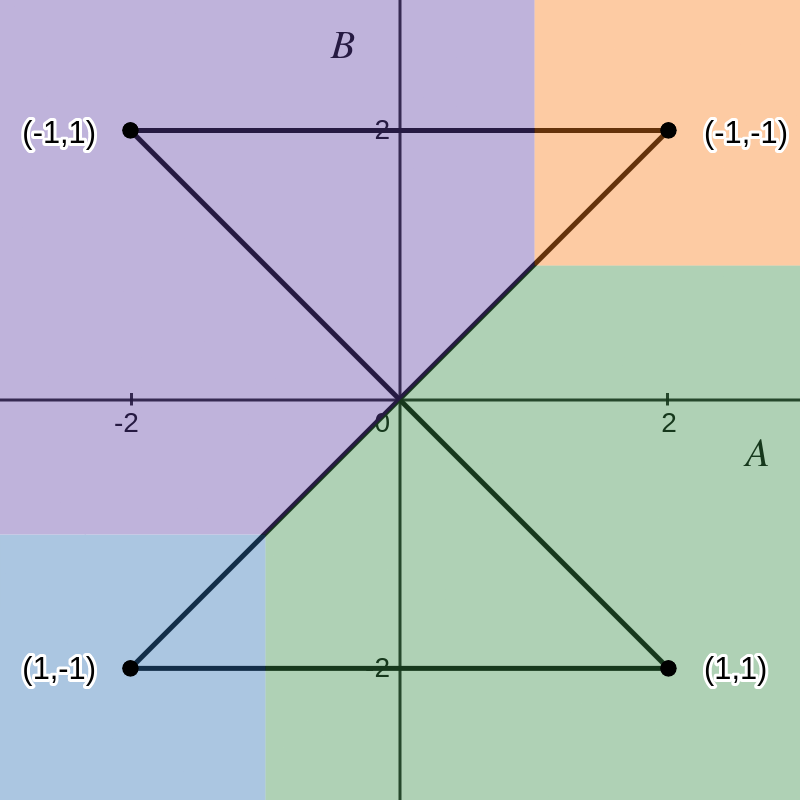}
    \caption{A visual ``proof" of the infeasibility of the circuit
    \begin{align}
        f(x_1, x_2) = (\textsf{XOR}(x_1,x_2),x_1).
    \end{align}
    Any map from $\Sigma^2$ to $a$-space that takes each input state to its corresponding minimizing cell would necessarily ``twist" the square, meaning it cannot be affine.}
    \label{fig:xor_x1}
\end{SCfigure}

\begin{theorem}[Feasibility in terms of the residual Hamiltonian]
    \label{thm:res_ham_feas}
    Given a fixed $J \in \R^{m\times m}_{UT}$, a circuit $(n,m,f)$ is feasible if and only if there exists $A \in \Aff(\R^n,\R^m)$ such that
    \begin{align}
        A(x) \in \mathcal{C}_J(f(x)), \quad \forall x \in \Sigma^n.
    \end{align}
    That is, a circuit is feasible if and only if there is an affine map taking each input state $x$ to the minimizing cell of its associated output state $f(x)$. Furthermore, the Ising Hamiltonian
    \begin{align}
        H(x,y) = A(x)\cdot y + y^\t Jy
    \end{align}
    encodes $(n,m,f)$.
    \begin{proof}
        Define $H$ as above.
        For each $x \in \Sigma^n$, we have
        \begin{align}
                H(x,y) > H(x, f(x)), \quad \forall y \in \Sigma^m\setminus\{f(x)\}
                &\iff E_J(A(x),y) > E_J(A(x),f(x)), \quad \forall y \in \Sigma^m\setminus\{f(x)\} \\
                &\iff M_J(A(x)) = f(x) \\
                &\iff A(x) \in \mathcal{C}_J(f(x)).
        \end{align}
    \end{proof}
\end{theorem}

\begin{SCfigure}[1]
    \centering
    \includegraphics[width=0.3\textwidth]{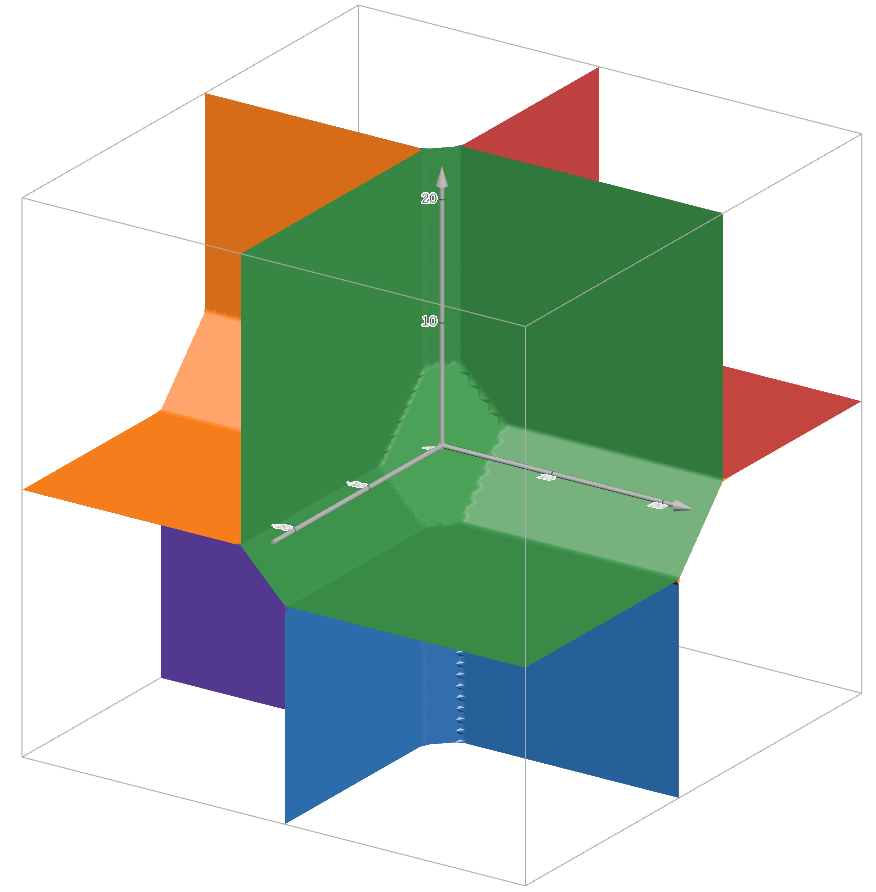}
    \caption{Visualization of the boundary set $B_J$ for a 3-bit-output circuit, with $J_{12} = 1, J_{13} = -2, J_{23} = 3$.}
    \label{fig:3_output}
\end{SCfigure}

\begin{figure}[h]
    \centering
    \includegraphics[width=0.25\textwidth]{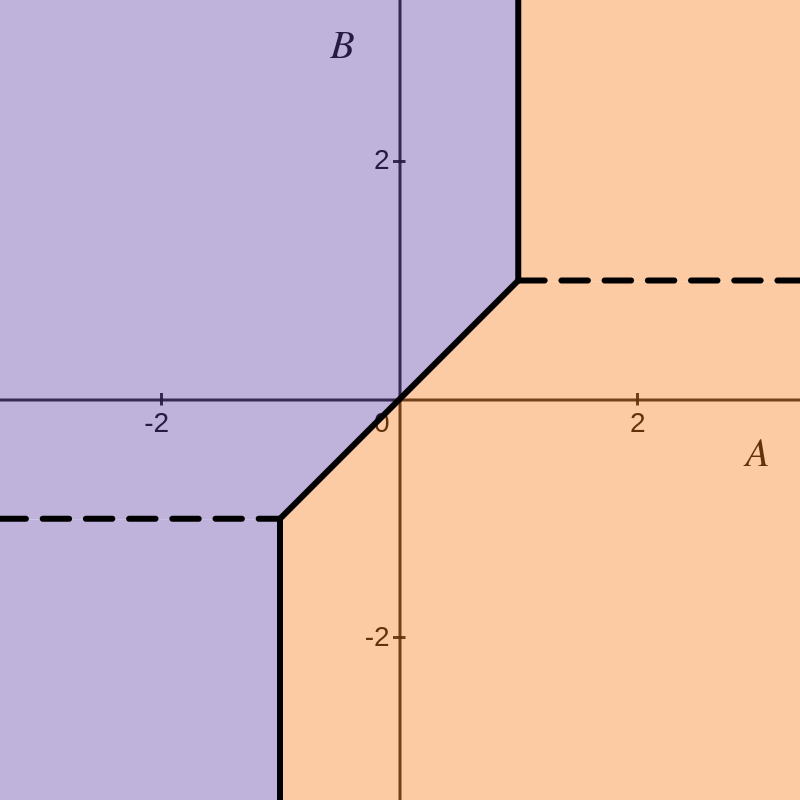}
    \quad
    \includegraphics[width=0.3\textwidth]{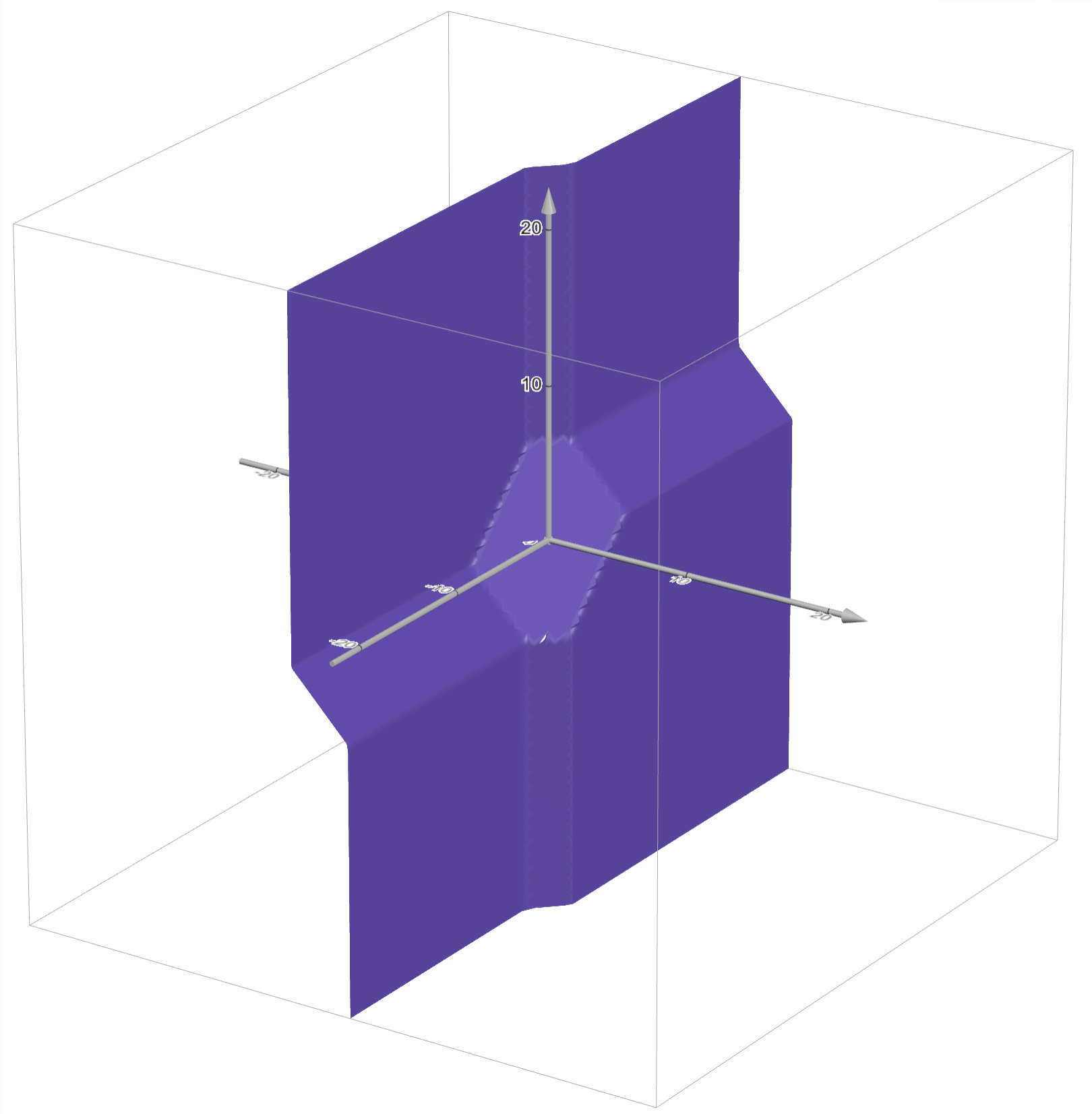}
    \caption{A visualization of $a$-space for shape $(n,1)$ circuits with one auxiliary (left) and two auxiliaries (right). The graphs are obtained by unioning minimizing cells that correspond to the same first output bit. 
    A circuit $(n,1,f)$ is feasible with one auxiliary if and only if there exists an affine map taking each $x \in \Sigma^n$ to the orange region if $f(x) = -1$ and to the purple region if $f(x) = 1$. 
    The same circuit is feasible with two auxiliaries if and only if there exists an affine map taking each input state $x$ to the right of the purple surface if $f(x) = -1$ and to the left if $f(x) = 1$.}
    \label{fig:feas_w_aux}
\end{figure}

\begin{remark}[Auxiliary Spins]
    Recall \Cref{prop:aux_are_outs} states that a circuit $(n,m,f)$ is feasible with $k$ auxiliaries if and only if there exists some auxiliary map $g: \Sigma^n \to \Sigma^k$ such that $(n,m+k,f\times g)$ is feasible.  Hence, we can often treat auxiliaries as outputs and vice versa. Additionally, the map $g$ is often not unique, and the choice of $g$ makes no difference with respect to feasibility.

    In terms of the residual Hamiltonian, treating some of the output bits as auxiliaries is equivalent to combining certain minimizing cells.
    For example, finding an Ising Hamiltonian that encodes $(n,1,f)$ with one auxiliary is equivalent to finding a $J \in \R^{2\times2}_{UT}$ and $A \in \Aff(\R^n,\R^2)$ such that
    \begin{align}
        A(x) \in \mathcal{C}_J((f(x), -1) \cup \mathcal{C}_J((f(x),1)), \quad \forall x \in \Sigma^n.
    \end{align}
    Taking the union of minimizing cells is reflective of the fact that we are agnostic about the value of the auxiliary bit, so long as the output bit is correct. See \Cref{fig:feas_w_aux} for a visualization of this fact.
\end{remark}

\section{Ising Circuits Superset Affine 1-NN Classifiers}
\label{sec:superset}

The residual Hamiltonian diagrams bear a strong resemblance to Voronoi diagrams, the decision partitions of 1-nearest-neighbor classifiers. This turns out not to be a coincidence, as certain Voronoi partitions can be transformed into residual Ising minimizing partitions by an affine map. 

\begin{definition}[Voronoi Diagram]
    For a countable subset $S \subseteq \R^n$, the \textit{Voronoi diagram} of $S$, denoted $\mathcal{V}_S$, is the partition of $\R^n$ into the \textit{Voronoi cells} 
    \begin{align}
        \mathcal{V}_S(p) := \{x \in \R^n : \|x - p\| < \| x - q \|,\ \forall q \in S \setminus \{p\}\}
    \end{align}
    for all $p \in S$; combined with a final element of the partition containing the remaining points of $X$, called the \textit{boundary set} of the diagram.  Each $\mathcal{V}_S(p)$ contains all the points in $\R^n$ that are closer to $p$ than any other point in $S$.  The boundary set contains the points in $\R^n$ that are equidistant from two or more points in $S$.
\end{definition}

\begin{figure}[h]
    \centering
    \includegraphics[width = 0.25\textwidth]{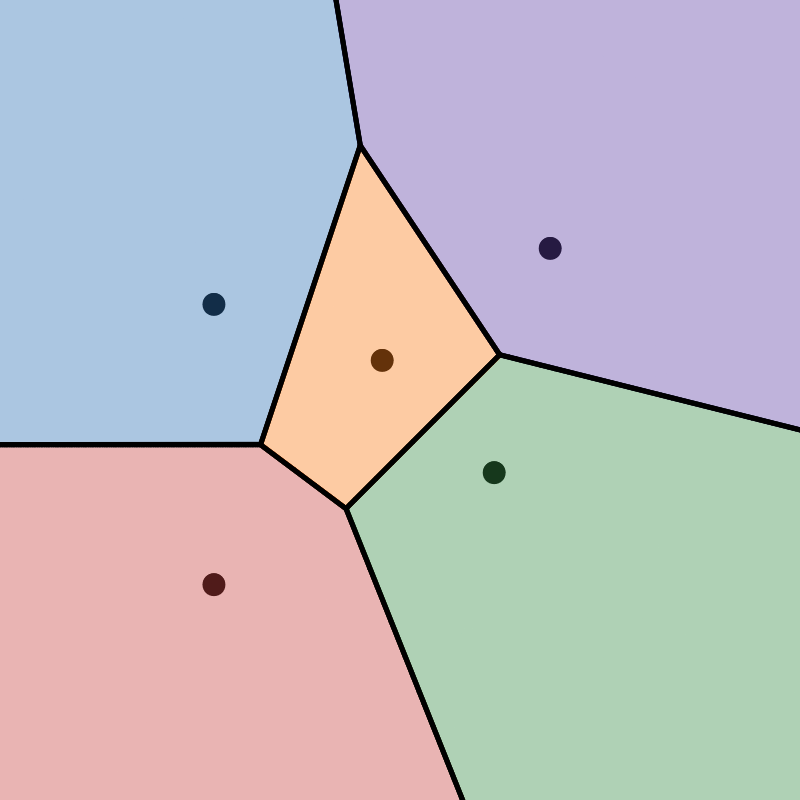}
    \caption{An example of a Voronoi diagram with $\#S = 5$.}
    \label{fig:voronoi}
\end{figure}

\begin{proposition}
    \label{prop:voronoipoly}
    For all $p \in S \subseteq \R^n$, $\mathcal{V}_S(p)$ is a convex polyhedron with H-representation\footnote{See the Wikipedia page \textit{Convex Polytope}.} $\{ P^S_{p,q} \}_{q \in S \setminus \{p\}}$, where
\begin{align}
    P^S_{p,q} := \left\{ x \in \R^n \middle| \langle p - q, x\rangle  + \frac{1}{2} (\|q\|^2 - \|p\|^2)  > 0 \right\}.
\end{align}
\begin{proof}
    Notice that
    \begin{align}
        \mathcal{V}_S(p) = \bigcap_{q\in S\setminus\{p\}} \{x \in \R^n: \|x - p\| < \|x - q\| \}.
    \end{align}
    For each $q \in S\setminus\{p\}$, we have
    \begin{align}
        \|x - p\| < \|x - q\|
        &\iff \|x - p\|^2 < \|x - q\|^2 \\
        &\iff \<x - p, x - p\> < \<x - q, x - q\> \\
        &\iff \|x\|^2 - 2\<p, x\> + \|p\|^2 < \|x\|^2 - 2\<q, x\> + \|q\|^2 \\
        &\iff -2\<p,x\> + 2\<q,x\> + \|p\|^2 - \|q\|^2 < 0 \\
        &\iff \<p - q, x\> + \frac{1}{2}(\|q\|^2 - \|p\|^2) > 0.
    \end{align}
\end{proof}
\end{proposition}

\begin{proposition}
\label{prop:isingpoly}
    Let $J \in \R^{m\times m}_{UT}$.  For all $y \in \Sigma^m$, the minimizing cell $\mathcal{C}_J(y)$ is a convex polyhedron with H-representation $\{Q^J_{y,z}\}_{z\in \Sigma^m\setminus\{y\}}$, where
    \begin{align}
        Q^J_{y,z} := \left\{ a \in \R^m \middle| \<a, z-y\>  + \<J, z^{\otimes2} - y^{\otimes2}\>_F > 0 \right\}.
    \end{align}
    \begin{proof}
        Pick any $y \in \Sigma^m$. We can rewrite the definition of $\mathcal{C}_J(y)$ as follows:
    \begin{align}
        \mathcal{C}_J(y) &= \{a \in \R^m : M_J(a) = y\} \\
        &= \{a \in \R^m : E_J(a, z) > E_J(a, y),\ \forall z \in \Sigma^m \setminus \{y\}\}\\
        &= \bigcap_{z \in \Sigma^m \setminus \{y\}} \{a \in \R^m : E_J(a, z) > E_J(a, y)\} \\
        &= \bigcap_{z \in \Sigma^m \setminus \{y\}} \{a \in \R^m : E_J(a, z) - E_J(a, y) > 0\}.
        \label{eq:intersectcj}
    \end{align}
    Now, observe that these expressions are affine functions of $A$:
    \begin{align}
        E_J(a, z) - E_J(a, y)
        &= a\cdot z + z^\t Jz - (a\cdot y + y^\t Jy) \\
        &= \<a,z\> + \<J, z^{\otimes2}\>_F - \<a,y\> - \<J, y^{\otimes2}\>_F \\
        &= \<a,z - y\> + \<J, z^{\otimes2} - y^{\otimes2}\>_F.
    \end{align}
    It follows that the intersection in Equation \eqref{eq:intersectcj} is an intersection of open half-spaces, and hence that $\mathcal{C}_J(y)$ is a convex polyhedron with H-representation given by these half spaces.
    \end{proof}
\end{proposition}

\begin{remark}
    Comparing \Cref{fig:voronoi} to \Cref{fig:a_space} and \Cref{prop:voronoipoly} to \Cref{prop:isingpoly}, we see that the minimizing partition $\mathcal{C}_J$ looks quite similar to a Voronoi diagram, both visually and algebraically. In this section, we define a new notion of feasibility in terms of Voronoi diagrams, and link it to our original notion in \Cref{thm:vor_feas->feas}.
\end{remark}

\begin{definition}
    A circuit $(n,m,f)$ is \textit{Voronoi feasible} if there exists $B \in \Aff(\R^m, \R^n)$ such that
    \begin{align}
        x \in \mathcal{V}_{B(\Sigma^m)}(Bf(x)), \quad \forall x \in \Sigma^n.
    \end{align}
    In other words, every input state $x \in \Sigma^n$ is closer to $Bf(x)$ than it is to $By$ for any $y \in \Sigma^m\setminus\{f(x)\}$. Such a $B$ is called a \textit{Voronoi solution}.
\end{definition}

\begin{definition}[Pseudo-adjoint]
    Given $A = [T, b] \in \Aff(\R^n, \R^m)$, we define its \textit{pseudo-adjoint} as the function $A^\adj(x) = T^\t(x - b)$. The pseudo-adjoint can also be written as the affine map $[T^\t, -T^\t b] \in \Aff(\R^m, \R^n)$. It has the useful property that $A^\adj A(x) = T^\t Tx$.  In some ways, $A^\adj$ behaves like an inverse of $A$.
\end{definition}

\begin{lemma}\label{lem:inj_vor_feas->feas}    
    Suppose that $B = [T, b] \in \Aff(\R^m, \R^n)$ is injective. Define $J_T = (T^\t T)^{(UT)}$.
    Then for any $y \in \Sigma^m$, 
    \begin{align}
        -B^\adj \left( \mathcal{V}_{B(\Sigma^m)}(By) \right) = \mathcal{C}_{J_T}(y).
    \end{align}
    \begin{proof}
        It is sufficient to show that $-B^\adj$ induces a bijection between the H-representations of $\mathcal{V}_{B(\Sigma^m)}(By)$ and $\mathcal{C}_{J_T}(y)$. Since $B$ is injective, it is a bijection $\Sigma^m \setminus \{y\} \lra B(\Sigma^m) \setminus \{By\}$. Therefore, it is sufficient to show that
        \begin{align}
           -B^\adj \left( P^{B(\Sigma^m)}_{By,Bz} \right) = Q^{J_T}_{y,z}, \quad \forall z \in \Sigma^m\setminus\{y\}.
        \end{align}
        Pick any such $z$, and pick any $x \in \R^n$. It is sufficient to show that 
        \begin{align}
            \< By - Bz, x\> + \frac{1}{2} (\|Bz\|^2 - \|By\|^2) = \<-B^\adj x, z-y \> + \<J_T, z^{\otimes2} - y^{\otimes2}\>_F.
        \end{align}
        We begin by applying the definition of $B$ and re-arranging to acquire the first term:
        \begin{align}
            \langle By - Bz, x\rangle + \frac{1}{2} (\|Bz\|^2 - \|By\|^2) 
            &= \langle Ty + b - Tz - b, x\rangle + \frac{1}{2} (\|Tz + b\|^2 - \|Ty + b\|^2)\\
            &= \langle y-z, T^\t x\rangle + \frac{1}{2} (\|Tz\|^2 + 2\langle Tz, b\rangle - \|Ty\|^2 - 2\langle Ty, b\rangle)\\
            &= \langle y-z, T^\t x - T^\t b\rangle + \frac{1}{2} (\|Tz\|^2 - \|Ty\|^2) \\
            &= \langle -B^\adj x, z-y \rangle + \frac{1}{2} (\|Tz\|^2 - \|Ty\|^2).
            \label{eq:interres}
        \end{align}
        Now, we analyze the second term. By \Cref{cor:frob_formulas},
        \begin{align}
            \|Tz\|^2 - \|Ty\|^2
            &= \<T^\t T, z^{\otimes2}\>_F - \<T^\t T, y^{\otimes2}\>_F \\
            &= \<T^\t T, z^{\otimes2} - y^{\otimes2}\>_F \\
            &= \sum_{j,k = 1}^m (T^\t T)_{jk} (z^{\otimes2} - y^{\otimes2})_{jk} \\
            &= \sum_{\substack{j,k=1, \\ j \ne k}}^m (T^\t T)_{jk} (z^{\otimes2} - y^{\otimes2})_{jk}
            + \sum_{j=1}^m (T^\t T)_{jj} (z^{\otimes2} - y^{\otimes2})_{jj}.
        \end{align}
        The second sum is zero since $z_j^2 = y_j^2 = 1$ for all $j$. By symmetry of $T^\t T$ and of $z^{\otimes2} - y^{\otimes2}$, the first sum is equal to
        \begin{align}
             2\sum_{j<k} (T^\t T)_{jk} (z^{\otimes2} - y^{\otimes2})_{jk}
             = 2\left\<(T^\t T)^{(UT)}, z^{\otimes2} - y^{\otimes2}\right\>_F
             = 2\<J_T, z^{\otimes2} - y^{\otimes2}\>_F.
        \end{align}
        Substituting into Equation \eqref{eq:interres}, we obtain
        \begin{align}
            \langle By - Bz, x\rangle + \frac{1}{2} (\|Bz\|^2 - \|By\|^2)
            = \langle -B^\adj x, z-y \rangle + \left\<J_T, z^{\otimes2} - y^{\otimes2}\right\>_F,
        \end{align}
        as required.
    \end{proof}
\end{lemma}


\begin{lemma}\label{lem:vor_feas->inj_vor_feas}
    If an Ising circuit is Voronoi feasible, then it has an injective Voronoi solution, meaning that $B$ is injective when restricted to $\Sigma^m$.
    \begin{proof}
        (Sketch). The proof closely follows the argument for \Cref{prop:generic}. We pick a solution $B$ and enumerate all degeneracies, i.e. $x \neq y$ such that $B(x) = B(y)$. Then we perturb $B$ very slightly to remove one degeneracy, where the perturbation is small enough that it cannot create new degeneracies, and points in $\Sigma^n$ stay in the same cell. Then iterate until we have a solution with no degeneracies, i.e. an injective Voronoi solution. Details are left to the reader. 
    \end{proof}
\end{lemma}

\begin{theorem}\label{thm:vor_feas->feas}
    If a circuit $(n,m,f)$ is Voronoi feasible, then it is feasible.
    Furthermore, if $B = [T,b]$ is an injective Voronoi solution, then the Hamiltonian
    \begin{align}
        H(x,y) = -B^\adj (x)\cdot y + \<J_T, y^{\otimes2}\>_F
    \end{align}
    encodes $(n,m,f)$, where $J_T = (T^\t T)^{(UT)}$.
    \begin{proof}
        By \Cref{lem:vor_feas->inj_vor_feas}, we know $(n,m,f)$ has an injective Voronoi solution $B = [T,b] \in \Aff(\R^m,\R^n)$.
        Let $H$ be defined as above. Then
        \begin{align}
            \text{$B$ is a Voronoi solution}
            &\iff x \in \mathcal{V}_{B(\Sigma^m)}(Bf(x)), \quad \forall x \in \Sigma^n \\
            &\iff -B^\adj (x) \in -B^\adj(\mathcal{V}_{B(\Sigma^m)}(Bf(x))), \quad \forall x \in \Sigma^n \\ 
            &\iff -B^\adj(x) \in \mathcal{C}_{J_T}(f(x)), \quad \forall x \in \Sigma^n
            & \text{(\Cref{lem:inj_vor_feas->feas})}\\
            &\iff \text{$H$ encodes $(n,m,f)$}
            & \text{(\Cref{thm:res_ham_feas})}.
        \end{align}
    \end{proof}
\end{theorem}

\section{Elimination of Local Minima}
\label{sec:localminima}

Like neural network optimization landscapes, the energy landscapes of Ising systems may be more or less regular. Landscapes which are quite noisy and spiky will clearly have different dynamics to those which are relatively convex, with a single deep minimum. To exert design authority over the shape of energy landscapes, we will begin by defining a precise way to model their qualitative structure in terms of graphs.

\begin{definition}[Hamming ball]
    For $x \in \Sigma^n$, the \textit{(closed) ball of radius $r$ centered at $x$} is
    \begin{align}
        B_r(x) = \{y \in \Sigma^n: d(x,y) \le r\},
    \end{align}
    where $d(x,y)$ is the Hamming distance between $x$ and $y$.
\end{definition}

\begin{definition}[Local and global minima]
    \label{def:loc_glob_min}
    Given an Ising Hamiltonian $H: \Sigma^n\times\Sigma^m \to \R$ and a fixed input state $x \in \Sigma^n$, a \textit{global minimum} of $H(x,\cdot)$ is an output state $y \in \Sigma^m$ that has lower energy than any other output state. More precisely, $y$ is a global minimum if and only if
    \begin{align}
        \label{eq:trad_constraints}
        \forall z \in \Sigma^m\setminus\{y\}: H(x,z) \ge H(x,y).
    \end{align}
    A \textit{local minimum} of $H(x,\cdot)$ is an output $y \in \Sigma^m$ that has lower energy than every output Hamming distance 1 from $y$. In other words,
    \begin{align}
        \forall z \in B_1(y)\setminus\{y\}: H(x,z) \ge H(x,y).
    \end{align}
    We say $y$ is a \textit{strict global/local minimum} if the above inequalities are strict.
    We define local and global minima of a residual Ising Hamiltonian $E_J(a,\cdot)$ similarly.
\end{definition}

\begin{proposition}
    \label{prop:global_min}
    An Ising Hamiltonian $H:\Sigma^n\times\Sigma^m \to \R$ encodes a circuit $(n,m,f)$ if and only if the output state $f(x)$ is a strict global minimum for every input state $x \in \Sigma^n$.
\end{proposition}

For the purposes of fine-tuning the set of inequality constraints, it is conceptually and mathematically useful to represent the set of linear inequalities in the form of graphs. 

\begin{definition}[Energy Graph of a Hamiltonian]
    Given an Ising Hamiltonian $H:\Sigma^n\times\Sigma^m \to \R$ and a fixed input $x \in \Sigma^n$, the \textit{energy graph} of $H(x,\cdot)$ is the directed graph $(\Sigma^m, \mathcal{E})$, where $(z,y) \in \mathcal{E}$ if and only if $H(x,z) > H(y,z)$. The energy graph of $H$ as a whole is the disjoint union of the energy graph of $H$ on every input level. 
\end{definition}

\begin{definition}[Energy Graph of a Constraint Set]
    A constraint set representing the problem of arranging energy levels may be written as $\{H(s) < H(s') | s, s'\in P\}$ where $P$ is some subset of $\Sigma^{n + m} \times \Sigma^{n + m}$. This can be represented as a digraph where the nodes are $\Sigma^{n+m}$ and the edges are $P$. 
\end{definition} 

In this language, we can reformulate the earlier setup of the reverse Ising problem as follows: ``A Hamiltonian solves a constraint set when the constraint set energy graph is a subgraph of the energy graph of a Hamiltonian.'' Furthermore, the constraints of \Cref{eq:linprog}, which we will call the traditional or global-minimum constraints, can be visualized as a disjoint union of $2^n$ star graphs on $2^m$ vertices each, as in \Cref{fig:trad_constraints}. However, these constraints do not completely determine the energy graph of the Hamiltonian, as demonstrated in \Cref{fig:local_min}. The same figure also demonstrates the possible existence of local minima other than the correct output state $f(x)$.
    
This poses a problem for dynamical realizations of Ising machines, either in software or hardware. Suppose $H:\Sigma^n\times\Sigma^m \to \R$ is an Ising Hamiltonian that encodes a circuit $(n,m,f)$. The na\"ive algorithm for using $H$ to compute $f$ is a sort of gradient-descent algorithm: Pick an $x \in \Sigma^n$ for which you wish to compute $f(x)$, and start at an arbitrary $y \in \Sigma^m$. Move from $y$ to the state $\tilde{y} \in B_1(y)$ that maximizes the decrease in energy $\Delta E = H(x,\tilde{y}) - H(x,y)$. Continue moving from state to state until moving no longer decreases the energy. A physical system with energy function $H$ which embodies low-temperature Glauber dynamics or Gibbs sampling of the Boltzmann probability distribution will act similarly, mostly rolling downhill, possibly with some noise. 
    
The hope is that after several iterations, the system will settle on correct output $f(x)$. This hope is achieved if there are no local minima (besides the global minimum $f(x)$). But if there are other local minima, and the temperature is low enough, the system may become trapped in the local minimum for a very long time. This problem is well known as the fundamental issue in non-convex optimization, and the typical approach is to employ an annealing schedule to allow noise to push the system out of false minima: by starting at high temperature and slowly cooling, chances are much better that the global minimum will be achieved. This slow cooling decreases the likelihood of the `trapped' or `glassy' state incurred by fast cooling. On the other hand, if we could remove local minima entirely, this would no longer be a concern, significantly simplifying the process of inference for Ising systems to a single imperative: just cool down as much as possible! Surprisingly, as we shall shortly see, this is actually possible to achieve without even sacrificing the linearity of the constraint set.

\begin{figure}[h]
    \centering
    \includegraphics[width=0.25\textwidth]{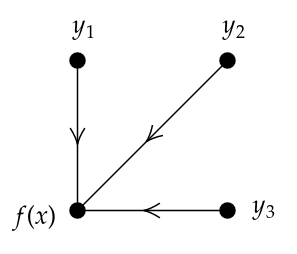}
    \caption{Global minimum constraints for $m = 2$. The state $f(x)$ is a strict global minimum if and only if all these edges are contained in the energy graph.}
    \label{fig:trad_constraints}
\end{figure}

\begin{figure}[h]
    \centering
    \includegraphics[width=0.25\textwidth]{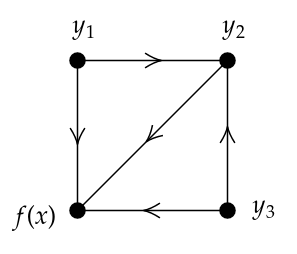}
    \includegraphics[width=0.25\textwidth]{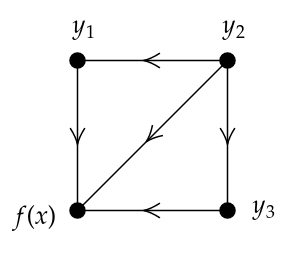}
    \caption{Two examples of energy graphs for $m = 2$ that satisfy the global minimum constraints. In the left graph, both $f(x)$ and $y_2$ are local minima, while in the right graph, only $f(x)$ is a local minimum.}
    \label{fig:local_min}
\end{figure}

\subsection{Residual Hamiltonian Approach}

    We can also understand local minima through the minimizing partition $\mathcal{C}_J$. Consider a point $a \in \R^m$ in the red region of \Cref{fig:loc_min_triangle} (left). Since $a \in \mathcal{C}_J((1,-1))$, the output $(1,-1)$ is a global minimum of $E_J(a,\cdot)$. However, $a$ is also in two of the three half spaces that define $\mathcal{C}_J((-1,1))$, namely,
    \begin{align}
        Q^J_{(-1,1),(1,1)} &= \left\{a \in \R^M \middle| E_J(a,(-1,1)) < E_J(a,(1,1))\right\}, \\
        Q^J_{(-1,1),(-1,-1)} &= \left\{a \in \R^M \middle| E_J(a,(-1,1)) < E_J(a,(-1,-1))\right\}.
    \end{align}
    This means that $(-1,1)$ is a local minimum of $E_J(a,\cdot)$. Likewise, if
    \begin{align}
        a \in \mathcal{C}_J((-1,1)) \cap Q^J_{(1,-1),(1,1)} \cap Q^J_{(1,-1),(-1,-1)},
    \end{align}
    then $(-1,1)$ is a global minimum and $(1,-1)$ is a local minimum. On the other hand, if $a$ avoids the red square in \Cref{fig:loc_min_triangle} (right), then there will be no local minima besides the global minimum.

    However, the subsets of $\mathcal{C}_J((-1,1))$ and $\mathcal{C}_J((1,-1))$ that lie outside the red square are not convex, which poses an issue for linear programming schemes. Hence, we take even smaller subsets of $\mathcal{C}_J((-1,1))$ and $\mathcal{C}_J((1,-1))$ that are convex, as seen in \Cref{fig:loc_min_triangle} (right). \Cref{thm:loc_min} generalizes this process to arbitrary $m$ and $J$.

\begin{figure}[h]
    \centering
    \includegraphics[width=0.25\textwidth]
    {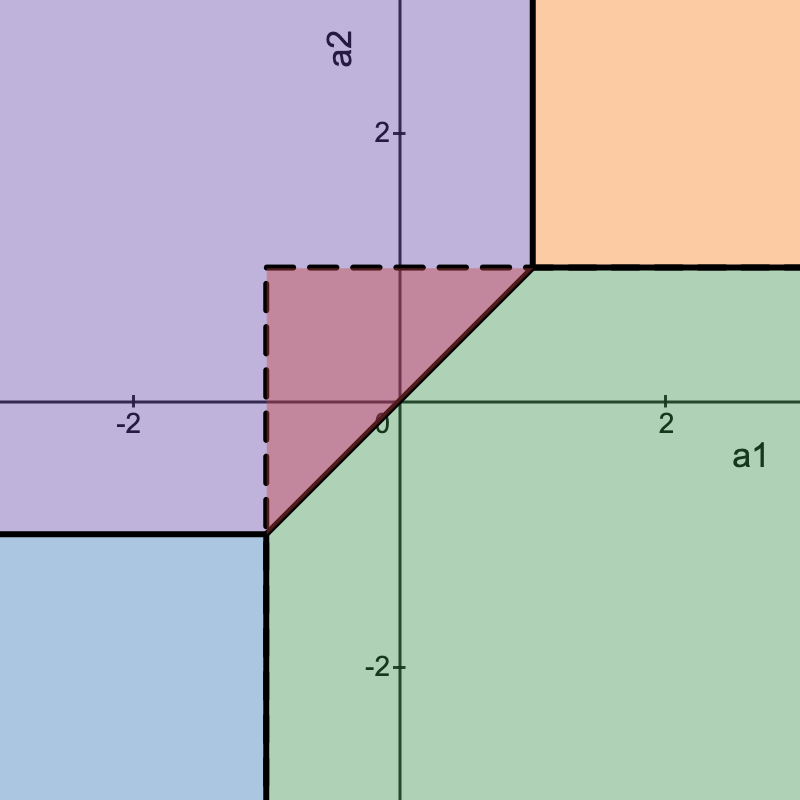}
    \hspace{3cm}
    \includegraphics[width=0.25\textwidth]
    {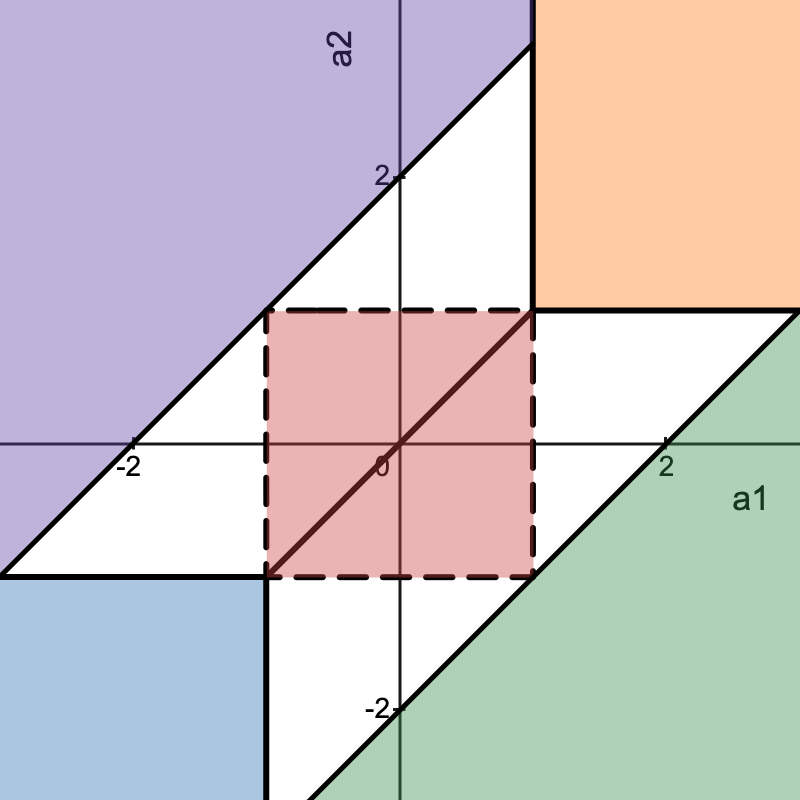}
    \caption{(Left) For $m = 2$ and $J_{12}=1$, highlighted in red, the region in $a$-space where $E_J(a,\cdot)$ has a global minimum of $(1,-1)$ and a local minimum of $(1,-1)$.
    (Right) The largest convex subsets of $\mathcal{C}_J((1,-1))$ (purple) and $\mathcal{C}_J((-1,1))$ (green) that avoid the region where $E_J(a,\cdot)$ has extraneous local minima (red).}
    \label{fig:loc_min_triangle}
\end{figure}

\begin{lemma}
    \label{lem:loc_min}
    Let $J \in \R^{m\times m}_{UT}$, and let $\tJ = \sym(J)$. Fix $a \in \R^m$. Then $y \in \Sigma^m$ is a local minimum of $E_J(a,\cdot)$ if and only if
    \begin{align}
        a_iy_i + 2\sum_{j=1}^m \tJ_{ij}y_iy_j \le 0, \quad \forall i = 1, \dots, m.
    \end{align}
    \begin{proof}
        Recall $\<\tJ,y^{\otimes2}\>_F = \<J,y^{\otimes2}\>_F$. For $y \in \R^m$ and for $i = 1, \dots, m$, define $\hy^i \in \R^m$ by
        \begin{align}
            \hy^i_j = \begin{cases}
                y_j, & j \ne i \\
                -y_j, & j = i
            \end{cases}.
        \end{align}
        We have that $y$ is a local minimum of $E_J(a,\cdot)$ if and only if
        \begin{align}
            \label{eq:loc_min}
            E_J(a,y) - E_J(a,\hy^i) \le 0, \quad \forall i = 1, \dots, m.
        \end{align}
        Notice
        \begin{align}
            E_J(a,y) - E_J(a,\hy^i)
            = \<a, y - \hy^i\> + \<\tJ, y^{\otimes2} - (\hy^i)^{\otimes2}\>_F.
        \end{align}
        The only nonzero component of $y - \hy^i$ is $(y - \hy^i)_i = 2y_i$, and the only nonzero components of $y^{\otimes2} - (\hy^i)^{\otimes2}$ are
        \begin{align}
            (y^{\otimes2} - (\hy^i)^{\otimes2})_{jk} = 2y_jy_k,
        \end{align}
        when exactly one of $j$ or $k$ is not equal to $i$. Therefore,
        \begin{align}
            E_J(a,y) - E_J(a,\hy^i)
            = 2a_iy_i + \sum_{\substack{j=1, \\ j\ne i}}^m \tJ_{ji}(2y_jy_i) + \sum_{\substack{k=1, \\ k\ne i}}^m \tJ_{ik}(2y_iy_k).
        \end{align}
        By symmetry of $\tJ$, both of the above sums are the same. Also, $\tJ_{jj} = 0$. Hence,
        \begin{align}
            E_J(a,y) - E_J(a,\hy^i)
            = 2a_iy_i + 2\sum_{j=1}^m \tJ_{ij}(2y_iy_j).
        \end{align}
        Thus, by \eqref{eq:loc_min}, $y$ is a local minimum if and only if
        \begin{align}
            \label{eq:loc_min_lemma}
            a_iy_i + 2\sum_{j=1}^m \tJ_{ij}y_iy_j \le 0, \quad \forall i = 1, \dots, m,
        \end{align}
    \end{proof}
\end{lemma}

\begin{remark}
    If we treat $E_J(a,\cdot)$ as a function on $\R^m$ instead of $\Sigma^m$, its partial derivative with respect to $y_i$ is
    \begin{align}
        \frac{\partial E_J}{\partial y_i}(a,y) = a_i + 2\sum_{j=1}^m \tJ_{ij}y_j.
    \end{align}
    Also, if $\hat{u}^i$ is the unit vector in the direction of $\hy^i - y$, then $\hat{u}^i_j = -y_i\delta_{ij}$. In light of this, the condition of \Cref{lem:loc_min} can be rephrased as
    \begin{align}
        \nabla_y E_J(a,y) \cdot \hat{u}^i \ge 0, \quad \forall i = 1, \dots, m.
    \end{align}
    In other words, $y \in \Sigma^m$ is a local minimum of $E_J(a,\cdot)$ if and only if the directional derivative of $E_J(a,\cdot)$ in the direction of every neighbor of $y$ is positive. This behavior is analogous to that of smooth functions over $\R^m$.
    
    This is somewhat surprising, since restricting functions over $\R^m$ to discrete lattices often changes their extrema in unpredictable ways. However, since $E_J(a,\cdot)$ is a multilinear polynomial, its restriction along any coordinate axis is an affine function, and thus the usual intuition about smooth functions holds.
\end{remark}

\begin{theorem}[Local minima]
    \label{thm:loc_min}
    Let $a \in \R^m$ and $J \in \R^{m\times m}_{UT}$. If $a$ and $J$ are such that
    \begin{align}
        E_J(a,z) - E_J(a,y) < E_J(0, z - y), \quad \forall y \in \Sigma^m\setminus\{z\},
    \end{align}
    then $E_J(a,\cdot)$ has no local minima apart from $z$.
    \begin{proof}
        Let $y, z \in \Sigma^m$ with $y \ne z$. Define
        \begin{align}
            S := \{i : y_i \ne z_i\} = \{i : y_i = -z_i\}.
        \end{align}
        Suppose $y$ is a local minimum of $E_J(a,\cdot)$. Then
        \begin{align}
            E_J(a,z) - E_J(a,y)
            &= \<a, z-y\> + \sum_{i,j} \tJ_{ij} z_iz_j - \sum_{i,j} \tJ_{ij} y_iy_j \\
            &= \sum_{i\in S} a_i(-2y_i) + \sum_{i,j} \tJ_{ij}(z_iz_j - y_iy_j) \\
            &\ge 4\sum_{i\in S, j} \tJ_{ij}y_iy_j + \sum_{i,j} \tJ_{ij}(z_iz_j - y_iy_j)
            & \text{(\Cref{lem:loc_min})}.
            \label{eq:loc_min_step}
        \end{align}
        Focusing on the first term, we see
        \begin{align}
            4\sum_{i\in S, j} \tJ_{ij}y_iy_j
            &= 4\sum_{i\in S} \left[ \sum_{j\in S} \tJ_{ij}y_iy_j + \sum_{j\in S^c} \tJ_{ij}y_iy_j \right] \\
            &= 4\sum_{i\in S, j\in S}\tJ_{ij}y_iy_j + 4\sum_{i\in S, j\in S^c}\tJ_{ij}y_iy_j \\
            &= 4\sum_{i\in S, j\in S}\tJ_{ij}z_iz_j - 4\sum_{i\in S, j\in S^c}\tJ_{ij}z_iz_j & (y_i = -z_i \text{ for } i \in S).
            \label{eq:loc_min_term1}
        \end{align}
        Let us return to the second term of \eqref{eq:loc_min_step}. Define
        \begin{align}
            \Delta_{ij} := \tJ_{ij}(z_iz_j - y_iy_j).
        \end{align}
        Then we have
        \begin{align}
            \sum_{i,j} \tJ_{ij}(z_iz_j - y_iy_j)
            &= \sum_{i,j} \Delta_{ij} \\
            &= \sum_{i\in S,j\in S}\Delta_{ij}
            + \sum_{i\in S^c,j\in S}\Delta_{ij}
            + \sum_{i\in S,j\in S^c}\Delta_{ij}
            + \sum_{i\in S^c,j\in S^c}\Delta_{ij} \\
            &= \sum_{i\in S,j\in S}0
            + \sum_{i\in S^c,j\in S}\tJ_{ij}(2z_iz_j)
            + \sum_{i\in S,j\in S^c}\tJ_{ij}(2z_iz_j)
            + \sum_{i\in S^c,j\in S^c}0 \\
            &= 4\sum_{i\in S,j\in S^c}\tJ_{ij}z_iz_j.
            \label{eq:loc_min_term2}
        \end{align}
        Adding \eqref{eq:loc_min_term1} and \eqref{eq:loc_min_term2}, we obtain
        \begin{align}
            E_J(a,z) - E_J(a,y)
            &\ge 4\sum_{i\in S, j\in S}\tJ_{ij}z_iz_j - 4\sum_{i\in S, j\in S^c}\tJ_{ij}z_iz_j + 4\sum_{i\in S,j\in S^c}\tJ_{ij}z_iz_j \\
            &= 4\sum_{i\in S, j\in S}\tJ_{ij}z_iz_j \\
            &= \sum_{i\in S,j\in S}\tJ_{ij}(z-y)_i(z-y)_j & (y_i = -z_i \text{ for } i \in S) \\
            &= \sum_{i,j} \tJ_{ij}(z-y)_i(z-y)_j \\
            &= E_J(0,z-y).
        \end{align}
        In summary, if $y$ is a local minimum of $E_J(a,\cdot)$, then
        \begin{align}
            E_J(a,z) - E_J(a,y) \ge E_J(0,z-y).
        \end{align}
        By contrapositive, if
        \begin{align}
            E_J(a,z) - E_J(a,y) < E_J(0,z-y), \quad \forall y \in \Sigma^m\setminus\{z\},
        \end{align}
        then $E_J(a,\cdot)$ has no local minima except $z$.
    \end{proof}
\end{theorem}
 
\begin{corollary}[Local minima]\label{cor:loc_min}
    Let
    \begin{align}
        H(x,y) = A(x)\cdot y +  \<J, y^{\otimes2}\>_F
    \end{align}
    be an Ising Hamiltonian that encodes the circuit $(n,m,f)$. For $x \in \Sigma^n$, if
    \begin{align}\label{eq:loc_min_constraints}
        H(x,f(x)) - H(x,y) < \<J, (f(x)-y)^{\otimes2}\>_F, \quad \forall y \in \Sigma^m\setminus\{f(x)\},
    \end{align}
    then $H(x,\cdot)$ has no local minima apart from $f(x)$.
    \begin{proof}
        Apply \Cref{thm:loc_min} with $a = A(x)$ and $z = f(x)$.
    \end{proof}
\end{corollary}

We see that requiring $H$ to be free of extraneous local minima is achievable by replacing the constraints on $H$ with slightly stricter versions. These more strict conditions can be incorporated into a linear program in the same manner as with the original constraints. Therefore, we can produce ``local-minima-free" Ising Hamiltonians with little added mental effort. Stricter constraints may require the addition of more auxiliary spins, increased dynamic range, etc. However, for small circuits where performance and reliability is important, a guarantee that there are no local minima is quite valuable, and may be worth the increased spin budget. 

\subsubsection{Spanning Tree Approach}

Suppose that we have a Hamiltonian $H$ and input state $x$ such that the energy graph of $H(x,\cdot)$ contains the edges shown in \Cref{fig:span-tree}. By the transitive property of inequality, such an energy graph must also contain the edge $(y_2,f(x))$. Thus, $H(x,\cdot)$ satisfies the requirement that $f(x)$ is a global minimum. Furthermore, every output state other than $f(x)$ has at least one edge flowing out of it. Thus, $H(x,\cdot)$ has no local minima besides $f(x)$. The graph in \Cref{fig:span-tree} is an example of a directed spanning tree of the hypercube graph.\footnote{Since the graph is directed, a more precise term for this structure would be \textit{arborescence} (see the Wikipedia page \textit{Arborescence (graph theory)}), but we will continue to use spanning tree.}

    \begin{figure}[h]
    \centering
    \includegraphics[width=0.25\textwidth]{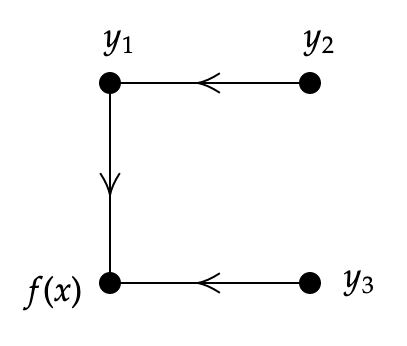}
    \caption{An example of an energy graph containing a spanning tree.  Transitivity of inequality implies the root of the tree, namely, $f(x)$ is a global minimum.}
    \label{fig:span-tree}
\end{figure}

    In fact, if the energy graph of $H(x,\cdot)$ contains a spanning tree with root $f(x)$, then $H(x,\cdot)$ necessarily has $f(x)$ as a global minimum and has no other local minima.  Therefore to require that a circuit have no spurious local minima it is sufficient to enforce that the energy graph of $H(x,\cdot)$ contains some spanning tree with root $f(x)$. In fact, we can make this statement even stricter:

    \begin{proposition}
        $H(x,\cdot)$ has no local minima if and only if there exists a directed spanning tree $T$ of the hypercube graph of $\Sigma^m$ such that $H(x,\cdot)$ satisfies the constraint set represented by $T$.  
        \begin{proof}
            The reverse direction is trivial, since as we have noted spanning tree graphs cannot have local minima. To prove the forward direction, we can construct a spanning tree graph by an iterative procedure. Begin by setting $T = (\{f(x)\}, \emptyset)$, i.e. the tree with only the root node and no edges. Now iterate over the nodes in $\Sigma^m \setminus \{f(x)\}$. For each such node, either it is in $V(T)$ or it is not. If is is, skip to the next node. Now, assume that it is not. Because $H(x, \cdot)$ has no local minima, for any node $y$ in $\Sigma^m$ such that $y \neq f(x)$, there exists a non-self-intersecting directed path $y \rightarrow f(x)$, which we can denote as $y = v_0, v_1, \dots, v_k = f(x)$, such that $H(x,v_i) \geq H(x, v_{i+1})$ and $d(v_i, v_{i+1}) = 1$ where $d$ denotes Hamming distance. Now consider the sub-path connecting $y$ to $T$ as follows: let $s = \min \{i : v_i \in V(T)\}$. Then add the path $v_0, \dots, v_s$ to $T$. Since each vertex in the chain except the last one is new to $T$, and the chain does not self-intersect, we have not added any cycles, so $T$ remains a directed tree. Iterating this process over all nodes must create a directed spanning tree of the hypercube graph which is a subgraph of the energy graph of $H(x,\cdot)$. 
        \end{proof}
    \end{proposition}

    This suggests that to design Hamiltonians without local minima, we should enforce the linear constraints corresponding to a directed spanning tree on each input level. However, it is not clear \textit{a priori} which spanning tree we should try to enforce, and the space of possible spanning trees is huge. Enforcing `bad' constraints makes the problem more difficult, either by rendering it infeasible, or by increasing the $L^1$ norm of the solution. 
    
    Fortunately, there is a simple way to iteratively improve the choices of spanning trees. Suppose that we start with a Hamiltonian $H_0$ solving the no-local-minima constraint in \Cref{cor:loc_min}. This guarantees that we have a solution without local minima. Now, suppose that we have a Hamiltonian $H_i$ without local minima. Then $H_i$ must satisfy a spanning tree on each input level. We construct a tree for each input level by selecting the out-flowing edge for each vertex which maximizes the energy drop:
    \begin{align}
        T_x = \left(\Sigma^m, \left\{\left(y,\argmin_{z : d(z, y) = 1} H_i(x,z) \right)\middle| y \in \Sigma^m \setminus \{f(x)\}\right\}\right)
    \end{align}
    Now, let $H_{i+1}$ be a solution to the linear program represented by the constraint sets $T_x$. We repeat until stability is reached. This procedure creates Hamiltonians with much lower $L^1$ norm than the original $H_0$, though there is no guarantee of uniqueness or optimality of the final solution.

\section{Acknowledgments}

This work was funded by the Laboratory for Physical Sciences (LPS) in Catonsville, Maryland. The authors would like to thank their team members John Daly, Maria Gorgone, and Teresa Ranadive for their frequent collaboration on this subject. 

\printnomenclature

\appendix 

\section{Linear Algebra Background}
\label{sec:linearalgebra}

\begin{definition}[Product function]
    Given functions $f_1:\R^n \to \R^{m_1}$ and $f_2:\R^n \to \R^{m_2}$, their \textit{product} is $f_1\times f_2: \Sigma^n \to \Sigma^{m_1+m_2}$ given by $(f_1\times f_2)(x) = (f_1(x), f_2(x))$.

    Given $f = f_1\times\dots\times f_m$, we call the functions $f_i$ the \textit{components} of $f$.
\end{definition}

\begin{definition}[Affine map]
    Let $V$ and $W$ be vector spaces.
    An \textit{affine map} from $V$ to $W$ is a function $A : V \lra W$ that can be written $A(x) = Tx + b$ for some linear function $T: V \to W$ and vector $b \in W$. We will write $A := [T, b]$. Despite the fact that $A$ may not be linear, we abbreviate $A(x)$ as $Ax$ for convenience. We will denote the set of all affine maps from $V$ to $W$ as $\Aff(V, W)$.

    In the case $V = \R^n$ and $W = \R^m$, $T$ is simply a matrix in $\R^{m\times n}$.
\end{definition}



\begin{definition}[Strictly upper-triangular matrix]
    A matrix $A \in \R^{n\times n}$ is \textit{strictly upper-triangular} if $A_{ij} = 0$ whenever $i \ge j$.  That is, $A$ has no nonzero entries on or below the diagonal.  The space of all such matrices is $\R^{n\times n}_{UT}$.

    The \textit{upper-triangular part} of a matrix $A \in \R^{n\times n}$ is $A^{(UT)}$, defined by
    \begin{equation}
        A^{(UT)}_{ij} = \begin{cases}
            A_{ij}, & i < j \\
            0, & i \ge j
        \end{cases}.
    \end{equation}
\end{definition}

\begin{definition}[Outer Product]
    Given $x \in \R^n$ and $y \in \R^m$, their \textit{outer product} is the matrix
    \begin{align}
        x \otimes y = xy^\t = (x_iy_j)_{1\le i\le n, 1\le j \le m} \in \R^{n\times m}.
    \end{align}
    We abbreviate $x\otimes x$ as $x^{\otimes2}$.
\end{definition}

\begin{definition}[Frobenius inner product]
    Given two matrices $A, B \in \R^{m\times n}$, their \textit{Frobenius inner product} is
    \begin{align}
        \<A, B\>_F = \sum_{i=1}^n \sum_{j=1}^m A_{ij}B_{ij}.
    \end{align}
    Note its similarity to the usual dot product of vectors.
\end{definition}

\begin{proposition}
    Let $x,y \in \R^n$ and $A, B \in \R^{m\times n}$. Then
    \begin{align}
        \<Ax, By\> = \<A^\t B, x\otimes y\>_F.
    \end{align}
    \begin{proof}
        \begin{align}
            \<Ax, By\>
            &= \sum_{i=1}^m (Ax)_i (Bx)_i \\
            &= \sum_{i=1}^m \left( \sum_{j=1}^n A_{ij}x_j \right) \left( \sum_{j=1}^n B_{ij}y_j \right) \\
            &= \sum_{i=1}^m \sum_{j=1}^n \sum_{k=1}^n A_{ij}B_{ik}x_jx_k \\
            &= \sum_{j=1}^n \sum_{k=1}^n x_j y_k \sum_{i=1}^m A_{ij}B_{ik} \\
            &= \sum_{j=1}^n \sum_{k=1}^n (x\otimes y)_{jk} (A^\t B)_{jk} \\
            &= \<A^\t B, x\otimes y\>_F.
        \end{align}
    \end{proof}
    \label{prop:frob_thm}
\end{proposition}

\begin{corollary}
    \label{cor:frob_formulas}
    For a vector $x \in \R^n$ and matrix $A \in \R^{n\times n}$, we have
    \begin{enumerate}
        \item
        \begin{align}
            x^\t Ax = \sum_{i,j = 1}^n A_{ij}x_ix_j = \langle A, x^{\otimes2} \rangle_F.
        \end{align}
        \item
        \begin{align}
            \|Ax\|^2 = \<A^\t A, x^{\otimes2} \>_F.
        \end{align}
    \end{enumerate}
    \begin{proof}
    $ $
        \begin{enumerate}
            \item
            \begin{align}
                x^\t Ax = \<Ix, Ax\> = \<I^\t A, x\otimes x\>_F = \<A, x^{\otimes2}\>_F.
            \end{align}
            \item
            \begin{align}
                \|Ax\|^2 = \<Ax, Ax\> = \<A^\t A, x\otimes x\>_F.
            \end{align}
        \end{enumerate}
    \end{proof}
\end{corollary}

\begin{proposition}[Symmetric Hollow Matrices]
    Let $A \in \R^{n\times n}$ be a nonzero symmetric hollow matrix.  That is,
    \begin{equation}
        A = \begin{pmatrix}
            0 & A_{12} & A_{13} & \dots \\
            A_{12} & 0 & A_{23} & \dots \\
            A_{13} & A_{23} & 0 & \dots \\
            \vdots & \vdots & \vdots & \ddots
        \end{pmatrix}.
    \end{equation}
    Then $A$ is indefinite, meaning $A$ has positive and negative eigenvalues.  Consequently, the function $x \mapsto x^\t A x$ is an indefinite quadratic form and has a saddle point.
    \begin{proof}
        Since $A$ is symmetric, its eigenvalues are real.  The sum of the eigenvalues of $A$ is equal to the trace of $A$, which is 0.  Thus, it only remains to prove the eigenvalues of $A$ are not all 0.  Since $A$ is real symmetric, it is diagonalizable, i.e. 
        \begin{equation}
            A = Q^{-1} \text{diag}(\lambda_1, \dots, \lambda_n)Q,
        \end{equation}
        for some matrix $Q$, where $\lambda_1, \dots, \lambda_n$ are the eigenvalues of $A$.  But if $\lambda_i = 0$ for all $i$, then $A = 0$, contrary to our assumption.
    \end{proof}
\end{proposition}

\section{Boolean Functions}
\label{sec:boolean}

The definitive reference for this section is the survey paper by Boros and Hammer \cite{Boros_2002}. See \cite{threshold_book} and \cite{threshold_msc_thesis} for specific literature on threshold functions.

\begin{definition}[Boolean function]
    A \textit{(vector-valued) Boolean function} is a function $f: \Sigma^n \to \Sigma^m$ or $f: B^n \to B^m$.
    We call the former convention \emph{spin} or \emph{multiplicative} convention and the latter \emph{Boolean} or \emph{additive} convention.
    Unless otherwise stated, we will always use the spin convention.

    A \textit{pseudo-Boolean function} is a function $f: \Sigma^n \to \R$.  We denote the set of all such functions as $\mathcal B^n$, and note that it is a $\R$-vector space.
\end{definition}

\begin{remark}[Converting between conventions]
    While we generally favor the spin convention over the Boolean convention, the two are completely equivalent, and we can translate between them at will.
    In particular, the degree of polynomials does not change when we translate between conventions, as the following example will illustrate.
    
    Define $\varphi: \R^n \to \R^n$ by $\varphi(x) = 2x - 1$.  Notice that $\varphi$ is an affine isomorphism from $\R^n$ to itself that takes $\Sigma^n$ to $B^n$.
    Furthermore, if $H: \Sigma^d \to \R$ is a quadratic polynomial, i.e.
    \begin{align}
        H(s) = b + \sum_i h_is_i + \sum_{i<j} J_{ij}s_is_j,
    \end{align}
    then there exists a quadratic polynomial $\widetilde{H}: B^n \to \R$ such that the diagram
    \begin{center}
        \begin{tikzcd}
            \Sigma^n \arrow[r, "H"] \arrow[d, "\varphi"] & \mathbb R \\
            B^n \arrow[ur, "\widetilde H"] &
        \end{tikzcd}
    \end{center}
    commutes.  Indeed,
    \begin{equation}
        \widetilde H(s) = \tilde b ~+~ \sum_{i}\tilde h_i s_i ~+~ \sum_{i<j}\tilde J_{ij}s_is_j,
    \end{equation}
    where
    \begin{itemize}
        \item $\tilde b = b + \sum\limits_{i<j}J_{ij} ~-~ \sum\limits_{i}h_i$,
        \item $\tilde h_i = 2h_i\left(\sum\limits_{\ell < i} J_{\ell i} ~+~ \sum\limits_{i < j} J_{ij}\right)$,
        \item $\tilde J_{ij} = 4J_{ij}$.
    \end{itemize}
    Viewing $H$ as an element of $S^{\leq 2}(\mathbb R) = \bigoplus_{i=0}^2 S^i(\mathbb R)$, the above definitions demonstrate that $\widetilde H$ is obtained from $H$ via an affine isomorphism from $S^{\leq 2}(\mathbb R)$ to itself.
\end{remark}

\begin{remark}[Another conversion between conventions]
    Notice that $\Sigma$ and $B$ are precisely the multiplicative and additive groups of order 2, and hence we have a group isomorphism $s \mapsto (-1)^s$ from $B$ to $\Sigma$.  This is different from the isomorphism $\varphi$ in the previous remark, as it takes $0$ to $1$ instead of $-1$.

    Note that if $\Sigma = \{-1,1\}$ is considered to be $\mathbb F_3^\times$, then $\Sigma^n$ is precisely the $\mathbb F_3$-rational points of the $n$-dimensional algebraic torus over $\mathbb F_3$. Alternatively, we may think of it as a subgroup of either $\mathbb R^\times$ or $\mathbb C^\times$. This perspective may be enlightening in the next section.
\end{remark}

\begin{definition}
    A \emph{threshold function} is a Boolean function $f: \Sigma^n \to \Sigma$ which can be written $f(x) = \operatorname{sgn}(w_0 + w\cdot x)$ for some $w_0\in \mathbb R$ and $w\in \mathbb R^n$. We write it this way so that in coordinates
    \begin{align}
        f(x) = \operatorname{sgn}(w_0 + w_1x_1 + ... + w_nx_n).
    \end{align}
    We call $w$ the \emph{weight vector} and $w_0$ the \emph{bias}.
\end{definition}

\begin{remark}
    Consider a threshold function $f(x) = \sgn(w_0 + w\cdot x)$.  Notice the equation $w_0 + w\cdot x = 0$ defines a hyperplane in $\R^n$. The function $f$ assigns $-1$ to every point in $\Sigma^n$ on one side of this hyperplane and $+1$ to every point on the other side.  Hence the name ``threshold function."

    Thus, a function $f: \Sigma^n \to \Sigma$ is a threshold function if and only if there is a hyperplane that separates $f^{-1}(\{-1\})$ from $f^{-1}(\{1\})$.
\end{remark}

\begin{remark}
    Threshold functions are well studied, although much is still unknown about them.  The number of $n$-dimensional threshold functions is only known up to $n = 9$.
\end{remark}

\begin{figure}[h]
    \centering
    \includegraphics[width=0.2\linewidth]{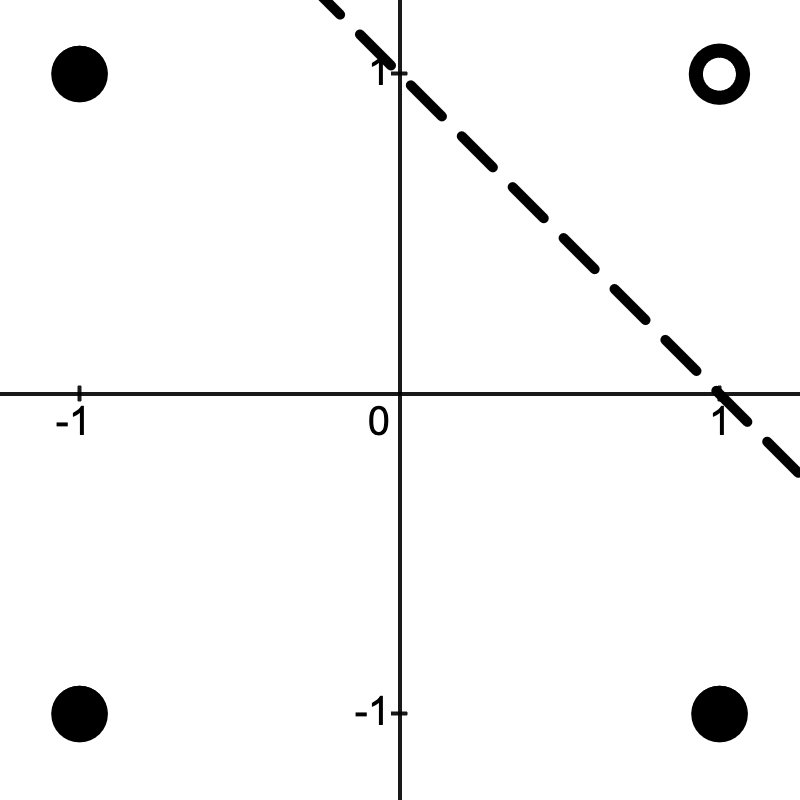}
    \quad \quad 
    \includegraphics[width=0.2\linewidth]{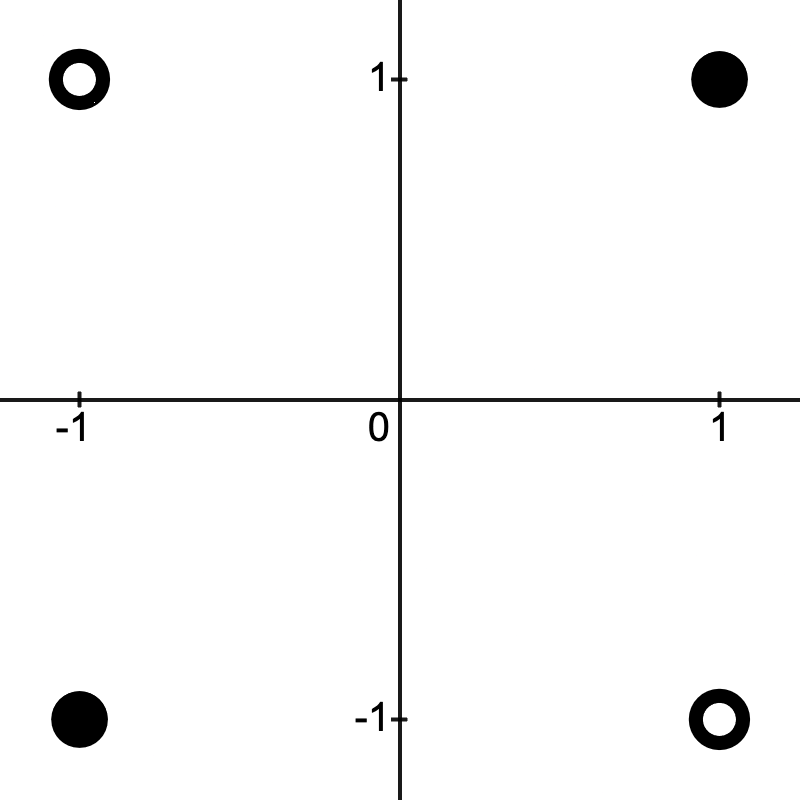}
    \caption{A visual depiction of the functions $f(x_1,x_2) = \AND(x_1,x_2)$ (left) and $f(x_1,x_2) = \XOR(x_1,x_2)$ (right).  The input space $\Sigma^n$ is plotted, with each input $x$ colored according to $f(x)$ (black for $-1$ and white for $1$).
    We see that $\AND$ is a threshold function, while $\XOR$ is not.}
    \label{fig:threshold}
\end{figure}

\section{Empirical Classification of Small Ising Circuits}
\label{sec:empirical}

\Cref{cor:threshold} can be viewed as a classification of shape $(n,1)$ circuits. General classification theorems for Ising circuits seem desirable, but as of now it is quite unclear what form such theorems should take. Nevertheless, in searching for such theorems, we made several observations regarding the structure of shape $(n,m)$ circuits for small values of $n$ and $m$. These are collected here.

We start with the following obvious lemma, which we include because it serves as the first case of every analysis to follow.

\begin{lemma}[Gluing]\label{lem:gluing}
    Let $(n,m_1,f_1)$ and $(n,m_2,f_2)$ be feasible Ising circuits.  Then $(n, m_1+m_2, f_1\times f_2)$ is feasible. 
    \begin{proof}
        Decompose $y \in \Sigma^{m_1+m_2}$ as $y = (y_1,y_2) \in \Sigma^{m_1}\times\Sigma^{m_2}$.
        Since $f_1$ and $f_2$ are feasible, they are encoded by Ising Hamiltonians $H_1(x,y_1)$ and $H_2(x,y_2)$.
        Then $H(x,y) := H_1(x,y_1) + H_2(x,y_2)$ is a quadratic pseudo-Boolean polynomial which $(n, m_1+m_2, f_1\times f_2)$. Additionally, we have
        \begin{align}
            H(x,y) = (A_1\times A_2)(x)\cdot y + y^\t
            \begin{pmatrix}
                J_1 & 0 \\
                0 & J_2
            \end{pmatrix} y.
        \end{align}
    \end{proof}
\end{lemma}

\subsection{Shape (1,m) and (2,m) Circuits}

\begin{proposition}
    All shape $(1,m)$ circuits are feasible.
    \begin{proof}
        The only functions $f: \Sigma \to \Sigma$ are
        \begin{enumerate}
            \item $f(x) = -1$,
            \item $f(x) = 1$,
            \item $f(x) = x$,
            \item $f(x) = -x$.
        \end{enumerate}
        These are all threshold functions and thus feasible by \Cref{cor:threshold}.  Consider a function $f = f_1\times\dots\times f_m : \Sigma \to \Sigma^m$.  Then each of its component functions $f_i:\Sigma\to\Sigma$ is feasible.
        Since $f$ consists of two independently feasible Ising circuits, it is feasible.
    \end{proof}
\end{proposition}

\begin{theorem}
    The only infeasible shape $(2,1)$ circuits are $\XOR$ and $-\XOR$.
    \begin{proof}
        All functions $f:\Sigma^2 \to \Sigma$ can be depicted visually as done in \Cref{fig:threshold}.  One can see that $\XOR$ and $-\XOR$ are the only ones that are not threshold functions.
    \end{proof}
\end{theorem}

\begin{theorem}
    A circuit $(2,m,f)$ is feasible if and only if at least one of the following is satisfied:
    \begin{enumerate}
        \item No component of $f$ is $\XOR$ or $-\XOR$.
        \item At least one component of $f$ is $\AND$, up to spin action.
    \end{enumerate}
    \begin{proof}
        If all component functions of $f$ are feasible, then $f$ is feasible by \Cref{lem:gluing}. Suppose then that $f$ has at least one infeasible component function, necessarily of shape $(2,1)$.

        The only infeasible shape $(2,1)$ circuit, up to sign, is $\XOR$. Such a circuit is solved by the presence of an $\AND$ component up to spin action equivalence \cite{andrewisaac}. Thus, a circuit is feasible either if it contains no $\XOR$ functions or if it contains at least one $\AND$ component up to linear equivalence. 
    \end{proof}
\end{theorem}

\subsection{Shape (3,m) Circuits}

\begin{definition}
    We call a circuit $(n,m,f)$ of \textit{type 0} if all of $f$'s components are feasible, of \textit{type 1} if some but not all $f$'s components are feasible, and of \textit{type 2} if none of  $f$'s components are feasible.
\end{definition}

\begin{proposition}$ $
    \begin{enumerate}
        \item All type 0 circuits are feasible.
        \item There exist feasible and infeasible type 1 circuits.
        \item There exist feasible and infeasible type 2 circuits.
    \end{enumerate}
    \begin{proof}$ $
        \begin{enumerate}
            \item This is a consequence of \Cref{lem:gluing}.
            \item The circuit $(\XOR\times 1)(x_1,x_2)$ is type 1 infeasible, and the circuit $(\XOR\times\AND)(x_1,x_2)$ is type 1 feasible.
            \item The circuit $(\XOR\times\XOR)(x_1,x_2)$ is type 2 infeasible, and it can be shown the circuit below is type 2 feasible.  It is written in Boolean convention for readability.
            \begin{equation}\begin{array}{ccc|ccc}
                x_1 & x_2 & x_3 & f_1 & f_2 & f_3 \\
                \hline
                0 & 0 & 0 & 0 & 0 & 0 \\
                0 & 0 & 1 & 0 & 0 & 0 \\
                0 & 1 & 0 & 0 & 0 & 0 \\
                0 & 1 & 1 & 0 & 1 & 1 \\
                1 & 0 & 0 & 0 & 0 & 0 \\
                1 & 0 & 1 & 1 & 0 & 1 \\
                1 & 1 & 0 & 1 & 1 & 1 \\
                1 & 1 & 1 & 0 & 0 & 0
            \end{array}\end{equation}
        \end{enumerate}
    \end{proof}
\end{proposition}

\begin{proposition}
    By exhaustive analysis of all  $m^{2^n} = 65,536$ shape $(3,2)$ circuits, we have determined that there are
    \begin{enumerate}
        \item 10,816 type 0 circuits, all of which are feasible,
        \item 31,616 type 1 circuits, 7,808 of which are feasible,
        \item 23,104 type 2 circuits, none of which are feasible.
    \end{enumerate}
\end{proposition}

\begin{remark}
    Of interest here is that there are no type 2 feasible shape $(3,2)$ circuits.  One might think that there are no type 2 feasible shape $(n,2)$ circuits, but this is incorrect, as the following is a counterexample:
    \begin{equation}\begin{array}{cccc|cc}
        x_1 & x_2 & x_3 & x_4 & f_1 & f_2 \\
        \hline
        0 & 0 & 0 & 0 & 0 & 0 \\
        0 & 0 & 0 & 1 & 0 & 0 \\
        0 & 0 & 1 & 0 & 0 & 0 \\
        0 & 0 & 1 & 1 & 0 & 0 \\
        0 & 1 & 0 & 0 & 0 & 0 \\
        0 & 1 & 0 & 1 & 0 & 0 \\
        0 & 1 & 1 & 0 & 0 & 0 \\
        0 & 1 & 1 & 1 & 0 & 1 \\
        1 & 0 & 0 & 0 & 0 & 0 \\
        1 & 0 & 0 & 1 & 0 & 0 \\
        1 & 0 & 1 & 0 & 0 & 1 \\
        1 & 0 & 1 & 1 & 0 & 1 \\
        1 & 1 & 0 & 0 & 0 & 0 \\
        1 & 1 & 0 & 1 & 1 & 0 \\
        1 & 1 & 1 & 0 & 1 & 0 \\
        1 & 1 & 1 & 1 & 0 & 1
    \end{array}\end{equation}
\end{remark}

\section{Tropical Curve Formulation}
\label{sec:tropical}
The cells in the Ising minimization partitions introduced in \Cref{sec:res_ham} can be realized as the connected components of the complement of a certain tropical polynomial. Similarly, it 

Fix $n$ and $s\in \Sigma^n$. Consider the following Laurent monomial in the variables $h_i$ and $J_{ij}$:
    \begin{align*}
        \varphi_s(h, J) = \prod h_i^{s_i} \cdot \prod_{i<j} J_{ij}^{s_i\cdot s_j}.
    \end{align*}
    By summing over all $n$ we obtain a polynomial which might be regarded as the "total Ising polynomial of size $n$":
    \begin{align*}
        \varepsilon(h,J) = \sum_{s\in \Sigma^n}\varphi^n_s(h,J).
    \end{align*}
    When the choice of $n$ is ambiguous, we will denote $\varphi_s$ by $\varphi^n_s$ and $\varepsilon$ by $\varepsilon^n$.
\begin{remark}
    The tropicalization $\operatorname{trop}(\varphi_s)$ of $\varphi_s$ is precisely an arbitrary Ising Hamiltonian on $n$ spins evaluated at the state $s$. The tropicalization $\operatorname{trop}(\varepsilon)$ of the characteristic Ising polynomial of size $n$ is the function
    \begin{align*}
        f(h,J) = \argmin_{s\in \Sigma^n} H_{h,J}(s)
    \end{align*}
    where $H_{h,J}$ is an Ising Hamiltonian on $n$ spins with local bias $h$ and interaction strength $J$.
\end{remark}
\begin{remark}
    The tropical variety $V(\trop \varepsilon)$ is precisely the set of Ising Hamiltonians on $n$ spins with degenerate ground states. Its complement $\mathbb R^{\frac{n(n+1)}{2}} \setminus V(\trop \varepsilon)$ has exactly $2^n$ connected components, and any two points in the same connected component have identical ground states. That is, there is a bijection $\Gamma:\Sigma^n \to H_0(\mathbb R^{\frac{n(n+1)}{2}} \setminus V(\trop \varepsilon)$ such that $\Gamma(s)$ is the simply connected region of all Ising Hamiltonians with nondegenerate ground state $s$.

    The argument in \Cref{lem:generic} could be simplified -- at the cost of obfuscation by general tropical nonsense -- by simply noting that any Hamiltonian with degenerate energy states must live on $V(\trop \varepsilon)$, which is codimension 1.
\end{remark}
\begin{definition}
    Fix $n$. Call $\varepsilon_J = \varepsilon(-, J)$ the ``linear piece of $\varepsilon$'' or the ``fixed $J$-polynomial''. This is precisely the \emph{residual Ising solution map} or \emph{ground state map} for fixed $J$ in equation (\ref{eqn:ground-state-map}). Likewise, we call $\varepsilon_h = \varepsilon(h, -)$ the ``fixed $h$-polynomial'' or the ``quadratic total polynomial''.
\end{definition}
\begin{lemma}
    $V(\operatorname{trop}(\varepsilon^n)) = V(\operatorname{trop}(\varepsilon^{n+1}_{h=0}))$.
\end{lemma}
\begin{proof}
    Let $s_1,...,s_n$ be the coordinates of $\Sigma^n$ and let $s_0$ be the additional coordinate in $\Sigma^{n+1}$. By identifying the variable $h_i$ in $\varepsilon^n$ with $J_{0,i}$ in $\varepsilon^{n+1}_{h=0}$, we see that these polynomials at least have the same number of variables.

    Notice that we are double counting monomials in $\varepsilon^{n+1}_0$. Indeed, 
    \begin{align}
        \operatorname{trop}(\varphi^{n+1}_s)(0, J) = \sum_{i<j} J_{i,j}\cdot s_is_j = \sum_{i<j} J_{i,j}\cdot (-s_i)(-s_j) = \operatorname{trop}(\varphi^{n+1}_{-s})(0,J).
    \end{align}
    Thus we discard one monomial for each antipodal pair $\varphi^{n+1}_{s}$ and $\varphi^{n+1}_{-s}$ in $\trop(\epsilon^{n+1}_{h=0}$ without affecting the value of $\trop(\epsilon^{n+1}_{h=0}$; in fact, a tropical polynomial with redundant terms is ill-posed. We can remove redundant terms by insisting that we only take $\varphi^{n+1}_{(1,s)}$ monomials such that $s_0 = 1$, but then
    \begin{align}
        \varphi^{n+1}_{(1,s)}(0,J) = \prod_{i<j} J_{i,j}^{s_is_j} = \prod_{i=1}^n J_{0,i}^{s_i} \cdot \prod_{1\leq i<j}J_{ij}^{s_is_j} = \varphi^n_s(h,J)
    \end{align}
    under the identification $J_{0,i}\mapsto h_i$.
\end{proof}

\printbibliography

\end{document}

%% file: header.tex
\usepackage{amssymb}
\usepackage{amsthm}
\usepackage[margin=1in]{geometry}
\usepackage[utf8]{inputenc}
\usepackage{mathtools}

\usepackage[singlelinecheck=false]{caption}
\usepackage{graphicx}
\usepackage{pgfplots}
\usepackage{pst-solides3d}
\usepackage{sidecap}
\usepackage{wrapfig}

\graphicspath{{./graphics/}}
\pgfplotsset{compat=1.18}
\sidecaptionvpos{figure}{m}

\usepackage{centernot}
\usepackage{extpfeil}
\usepackage{extarrows}
\usepackage{faktor}
\usepackage{mathrsfs}
\usepackage{xcolor}

\usepackage{hyperref}
\usepackage{cleveref}

\usepackage{tikz}
\usetikzlibrary{arrows}
\usetikzlibrary{automata}
\usetikzlibrary{calc}
\usetikzlibrary{cd}
\usetikzlibrary{decorations.markings}
\usetikzlibrary{decorations.pathmorphing}
\usetikzlibrary{shapes.geometric}

\tikzset{if/.code n args={3}{\pgfmathparse{#1}%
  \ifnum\pgfmathresult=1\pgfkeysalso{#2}\else\pgfkeysalso{#3}\fi}}
\tikzset{>=latex}
\tikzset{snake it/.style={decorate, decoration=snake}}
\tikzset{dbl/.style={>=stealth,
                     double,
                     double equal sign distance,
                     -implies,
                     shorten >=3pt,
                     shorten <=3pt}}          
\tikzstyle surjective=[postaction={decorate,decoration={markings, mark=at position .9 with {\arrow{latex}}}}]

\def\t{\top}
\def\adj{*}

\def\tJ{\tilde{J}}
\def\hy{\hat{y}}

\newcommand{\AND}{\textsf{AND}}

\newcommand{\XOR}{\textsf{XOR}}

\newcommand{\R}{\mathbb{R}}

\def\<{\langle}
\def\>{\rangle}

\newcommand{\lra}{\longrightarrow}

\newcommand\Aff{\text{Aff}}

\DeclareMathOperator*{\argmin}{argmin}

\DeclareMathOperator*{\sgn}{sgn}
\DeclareMathOperator*{\sym}{Sym}
\newcommand{\trop}{\operatorname{trop}}

\makeatletter
\newcommand{\tpitchfork}{%
  \vbox{
    \baselineskip\z@skip
    \lineskip-.52ex
    \lineskiplimit\maxdimen
    \m@th
    \ialign{##\crcr\hidewidth\smash{$-$}\hidewidth\crcr$\pitchfork$\crcr}
  }%
}
\makeatother

\theoremstyle{definition}
\newtheorem{theorem}{Theorem}[section]
\newtheorem{corollary}[theorem]{Corollary}
\newtheorem{definition}[theorem]{Definition}

\newtheorem{lemma}[theorem]{Lemma}
\newtheorem{proposition}[theorem]{Proposition}

\newtheorem{remark}[theorem]{Remark}